\documentclass[a4paper, 12 pt, twoside, reqno]{amsart}

\usepackage[english]{babel}
\usepackage[utf8]{inputenc}

\usepackage{setspace}
\usepackage[euler-digits,euler-hat-accent]{eulervm}

\usepackage{times}
\usepackage{amsmath}
\usepackage{amsfonts}
\usepackage{amssymb}
\usepackage{amsmath}
\usepackage{amsthm}
\usepackage{graphicx}
\usepackage{array}
\usepackage{color}
\usepackage{mathrsfs}
\usepackage{hyperref}
\usepackage{eucal}
\usepackage{esint} 
\usepackage{tikz}
\usepackage{upgreek}
\usepackage{enumitem}

\usepackage{tikz-cd}

\allowdisplaybreaks

\theoremstyle{plain}
\newtheorem{theorem}{Theorem}[section]
\newtheorem{corollary}[theorem]{Corollary}
\newtheorem{proposition}[theorem]{Proposition}
\newtheorem{lemma}[theorem]{Lemma}

\theoremstyle{definition}
\newtheorem{definition}[theorem]{Definition}

\theoremstyle{remark}
\newtheorem{remark}[theorem]{Remark}

\numberwithin{equation}{section}
\numberwithin{figure}{section}
\numberwithin{table}{section}

\newcommand{\R}{\mathbb{R}}
\newcommand{\N}{\mathbb{N}}
\newcommand{\C}{\mathbb{C}}                           

\newcommand{\Z}{\mathbb{Z}}

\newcommand{\s}[1]{\CMcal{#1}}
                  
\newcommand{\bb}[1]{\mathscr{#1}}
\newcommand{\rr}[1]{\mathfrak{#1}}
\newcommand{\n}[1]{\mathbb{#1}}

\newcommand{\expo}[1]{\,\mathrm{e}^{#1}\,}

\newcommand{\dd}{\,\mathrm{d}}
\newcommand{ \ii}{\,\mathrm{i}\,}

\newcommand{\virg}[1]{\lq\lq#1\rq\rq}                \newcommand{\ie}{\textsl{i.\,e.\,}}
\newcommand{\eg}{\textsl{e.\,g.\,}}
\newcommand{\cf}{\textsl{cf}.\,}
\newcommand{\etc}{\textsl{etc}.\,}

\newlength{\dhatheight}

\hypersetup{
pdftoolbar=true,        
pdfmenubar=true,        
pdffitwindow=true,     
pdfstartview=true,    
pdftitle={Topological Phases of the Weyl C*-algebra},    
pdfauthor={Giuseppe De Nittis, Santiago G. Rendel},     
breaklinks=true, 
colorlinks=true,       
linkcolor=purple,         
citecolor=teal, 
urlcolor=blue, 
bookmarksopen=true, 
filecolor=magenta,      
}

\begin{document}


\title[A scheme for topological phases of the Weyl \texorpdfstring{$C^*$}{C*}-algebra]{A scheme for topological phases of the Weyl \texorpdfstring{$C^*$}{C*}-algebra}


\author[G. De~Nittis]{Giuseppe De Nittis}

\address[G. De~Nittis]{Facultad de Matemáticas \& Instituto de Física,
  Pontificia Universidad Católica de Chile,
  Santiago, Chile.}
\email{gidenittis@uc.cl}

\author[S.~G. Rendel]{Santiago G. Rendel}

\address[S.~G. Rendel]{Department of Mathematics, KU Leuven, Leuven, Belgium.}
\email{santiago.g.rendel@kuleuven.be}

\vspace{2mm}

\date{\today}

\begin{abstract}
In this work, we introduce a classification scheme for topological phases of matter based on the topology of the space of pure states of a model $C^*$-algebra. Under it, topological phases are described by homotopy classes of sections of certain fiber bundles of (pure) states. Applying this classification procedure on states of the Weyl $C^*$-algebra that are invariant under translations by a lattice, we recover the $K$-theoretic classification of gapped spectral projectors for topological insulators of types A and AI, thus essentially generalizing this notion.

\medskip

\end{abstract}

\maketitle

\tableofcontents

\section{Introduction}\label{sec:Intr0}

The exploration and classification of topological phases of matter have become a vibrant area of research within mathematical physics, particularly due to their deep interplay between geometry, topology, and concrete problems in condensed matter physics. The story begins with the discovery of topological features in the quantum Hall effect (QHE) in the 1982 work \cite{thouless-kohmoto-nightingale-nijs-82}. These topological properties corresponded to topologically protected phases, and are the defining property of the systems now known as \emph{topological insulators} (TIs). 
The first unified framework for classifying these topological phases was Kitaev’s visionary proposal in \cite{kitaev-09} using $K$-theory, now widely known as Kitaev’s periodic table. Since then, $K$-theory has become a cornerstone in the classification of TIs, with its techniques extensively developed and generalized \cite{prodan-schulz-baldes-book}.

\medskip

However, one major shortcoming of Kitaev’s periodic table is its confinement to non-interacting systems. This limitation becomes evident in light of rigorous advancements establishing the QHE in interacting fermionic systems \cite{jaksic-ogata-pillet-06,hastings-michalakis-15,giuliani-mastropietro-porta-17,monaco-teufel-19,bachmann-bols-roeck-fraas-20}. These developments highlight the need for a broader framework, capable of capturing topological phases beyond the non-interacting case. The challenge lies in the fact that $K$-theory is ultimately used as a classification scheme for \emph{gapped spectral projectors} of Hamiltonians; objects that may not be well-defined or suitable in the presence of interactions. In extended systems, where the spectral projectors picture becomes murkier, the concept of a state provides a more natural foundation. This shift of perspective from classifying projectors to classifying states demands the development of new mathematical structures and equivalence principles. While there has been recent progress in this direction \cite{bhmpqs-24,spiegel-pflaum-25,artymowicz-kapustin-24}, a fully-fledged theory is still under construction. Many key open questions remain, and in this work we will focus on advancing the answer to a few of them, such as how to define equivalence of states appropriately, what classes of states admit a topological classification, and whether traces of $K$-theory can be recovered in a state-based framework.

\medskip

Our notion of topological phases based on states can be roughly summarized as follows. First, let $\bb{A}$ be a $C^*$-algebra containing the relevant physical observables, and consider a compact, Hausdorff space $X$, which we interpret as a \emph{quantum parameter space}. In many application the space $X$ emerges as (a compactification of) the spectrum of a symmetry group $\n{G}$ of the system. In this case the points of $X$ are interpreted as physically relevant quantum numbers associated with the symmetry. A typical example is provided by the group of \emph{spatial translations} with \emph{momentum} as the associated quantum number.
Let $\s{P}_\bb{A}$ the space of pure states of $\bb{A}$ topologized with the weak-$\ast$ topology, and $\s{Q}_\bb{A}\subseteq\s{P}_\bb{A}$ a relevant subspace, called the \emph{sample (state) space}, or \emph{state type}, defined by certain properties like the invariance under the action of a symmetry group $\n{G}$, or given structural conditions (\eg regularity, normality, complete factorization, \etc). 
A continuous function $X\to \s{Q}_\bb{A}$ is interpreted as a \emph{configuration} (of sample states) of the system. 
In many relevant situations, the sample space $\s{Q}_\bb{A}$ can show the structure of a \emph{fiber bundle} over $X$. In this case, one says that a configuration $F:X\to\s{Q}_\bb{A}$ is \emph{localizable} if it corresponds to a  {continuous section} of $\s{Q}_\bb{A}\to X$.
Let ${\rm Sec}(X,\s{Q}_\bb{A})$ be the space of such sections.
The topological phases of the system are then defined as the homotopy classes 
\begin{equation}\label{eq:itro_1}
\Omega(\s{Q}_\bb{A})\;:=\;\big[{\rm Sec}(X,\s{Q}_\bb{A})\big]\;.
\end{equation}
The notion of homotopy used here is stronger than usual. Concretely, let $\s{F}\subseteq C(Y,Z)$ be some family of continuous functions $\s{F}\subseteq C(Y,Z)$ between topological spaces $Y,Z$. We say that two elements in $\s{F}$ are homotopic if there exists a (usual) homotopy $H:[0,1]\times Y\to Z$ joining them, such that $H(t,\cdot)\in \s{F}$ for all $t\in[0,1]$. We denote the equivalence classes of $\s{F}$ under this notion of homotopy by $\left[\s{F}\right]$. Therefore, Equation \ref{eq:itro_1} means that two sections correspond to the same topological phase of the system if and only if they can be deformed one into the other via a homotopy $H:[0,1]\times X\to \s{Q}_\bb{A}$ such that $H(t,\cdot)$ is itself a section of $\s{Q}_\bb{A}$ for each $t\in[0,1]$. Additionally, we use the standard notation $[Y,Z]$ to mean homotopy classes of continuous maps from $Y$ to $Z$, and note that in fact $[Y,Z] := \left[ C(Y,Z) \right]$.

\medskip

In \cite{denittis-25} it was shown that this notion of topological phases can recover many of the known results about the classification 
of topological phases of non-interacting systems. This work will expand the scope for the application of this classification strategy. The key idea of the proposed framework is to relate the classification of (physical) topological phases with homotopy classes of sections of states, a point of view closely related in spirit with \emph{Kitaev's conjecture}, see \cite{bhmpqs-24,kubota-25} and the references therein. This now begs the question: When exactly can one apply this classification procedure? At the very least, the sample space of interest must have a fiber bundle structure to be able to define sections. Also, since topological phases typically appear in systems with symmetries, one should also expect the sample space to be invariant under the relevant symmetry group as well. The upcoming work \cite{rendel-25} will deal with answering this question in great generality, while in this article, we aim to apply the classification scheme in a specific physically relevant case: the \emph{Weyl} $C^*$-\emph{algebra}. 

\medskip

As such, it is useful to outline the general strategy that will be used in the following as a recipe.
It consists of the following steps:
\begin{itemize}
\item[i)] Fix the (unital) $C^*$-algebra of observables $\bb{W}$, and the abelian (topological) group $\n{G}$ of symmetries to consider.\vspace{1mm}
\item[ii)] Define the relevant sample space $\s{Q}_\bb{W}^\n{G}\subseteq \s{P}_\bb{W}$ related to the given symmetries $\n{G}$ of $\bb{W}$.
\vspace{1mm}
\item[iii)] Prove that the $C^*$-subalgebra invariant under the symmetries $\bb{V}_\n{G}$ is a $C^*$-bundle with constant fiber $\bb{O}$ over a compact space $X$ (the quantum parameter space).
\vspace{1mm}
\item[iv)] Show that the space of pure states $\s{P}_{\bb{V}_\n{G}}$ of $\bb{V}_\n{G}$ has the structure of a fiber bundle over $X$ with constant fiber $\s{P}_\bb{O}$. 
\vspace{1mm}
\item[v)] Observing that the restriction of pure invariant states of $\bb{W}$ yields pure states of $\bb{V}_\n{G}$, infer that
the sample space $\s{Q}_\bb{W}^\n{G}$ is the total space of a fiber bundle over $X$.
\vspace{1mm}
\item[vi)] Apply the classification procedure as in \eqref{eq:itro_1}.
\end{itemize}

\begin{remark}[Beyond pure states]
As will be discussed in Section \ref{sec:degeneracy}, one may also study topological phases associated with sample spaces of states that are no longer pure. We will study two cases: When we allow a fixed finite degeneracy for the states, we again obtain interesting topological phases. On the other hand, if we allow any possibly infinite degeneracy, which should be interpreted as having systems with non-zero temperature, then the topological phases become trivial.
\hfill $\blacktriangleleft$
\end{remark} 

\begin{remark}[A structural implication]
In \cite{rendel-25}, it will be proven that, even in some more general settings, step (iv), and in turn step (v), always follow if one has proven step (iii), so these steps become \virg{automatic}. However, in this work we will not use this general result and prove directly that (iv) and (v) hold in the cases of interest. 
\hfill $\blacktriangleleft$
\end{remark}

\begin{remark}[Quantum parameter space and group action]
As mentioned before, $X$ can often take the form of a compactification of the dual group of $\n{G}$. However, not any compactification of the dual group can take this role in a given algebra: It is not true in general that for any compactification $Y$ of $\widehat{\n{G}}$, the fixed-point subalgebra can be given the structure of a $C(Y)$-algebra. On the other hand, if the action is inner and $\n{G}$ is discrete, then the map $g\mapsto u_g$, with $u_g\in \bb{U}(\bb{W})$ implementing the action by $g\in \n{G}$, extends to a structure homomorphism $C(\widehat{\n{G}})\simeq C^*(\n{G}) \hookrightarrow \bb{Z}(\bb{V}_\n{G})$ into the center of the invariant subalgebra.
This discussion will also be covered in more detail in \cite{rendel-25}.
\hfill $\blacktriangleleft$
\end{remark}

\medskip

Concretely, the goal of this paper is to apply this program on the three sample spaces of interest covered in our previous work \cite{denittis-rendel-25}. First of all we consider  pure states  of the Weyl algebra invariant under continuous translations $\s{P}^{\tau}_\bb{W}$, with a special emphasis in the subspace of semi-regular states $\s{P}^{\tau,\beta}_\bb{W}$ (plane wave states). The second sample space $\s{P}^{\Gamma,\beta}_\bb{W}$ is made by pure, semi regular states of the Weyl algebra
invariant under translations by the lattice $\Gamma$ (Bloch wave states), $\s{P}^{\Gamma,\beta}_\bb{W}$. Finally, we will  consider the space $\s{P}_\bb{W}^{\rm Z}$ of pure states invariant under both spatial translations by $\Gamma$ and momentum translations by its dual lattice $\Gamma'$ (Zak states). These families of states were originally described as sets in \cite{beaume-manuceau-pellet-sirugue-74}. All of these classes of states are trivial fiber bundles over the spectrum of their symmetry groups.
The most interesting among them are Bloch wave states, which as we prove in \cite[Theorem 2.5.1]{denittis-rendel-25}, take the form of a fiber bundle over the Brillouin zone $\n{B}_\Gamma \simeq \n{T}^d$ with typical fiber ${^{\rm w}}\bb{G}_1$, the Hilbert Grassmannian of rank-$1$ orthogonal projections on a separable Hilbert space equipped with the weak operator topology.

\medskip

In this spin-less case, we obtain no non-trivial phases associated to states that are invariant under all translations, but upon including spin into our $C^*$-algebra of observables (by tensoring with a \virg{local} $C^*$-algebra encoding spin), we observe the expected topological phases for a homogenous system with spin. On the other hand, we recover the usual topological classification for topological insulators of types A and AI when studying the configurations associated to Bloch wave states. The classification of topological insulators of type A (Theorems \ref{theo_main_top} and \ref{theo_main_top-2}) is obtained from classifying the physically relevant configurations of the sample space $\s{P}^{\Gamma,\beta}_\bb{W}$, corresponding to its continuous sections $\n{B}_\Gamma\to \s{P}^{\Gamma,\beta}_\bb{W}$. On the other hand, the classification for topological insulators of type AI (Theorem \ref{theo_main_top-3}) is obtained by restricting to sections defining time-reversal invariant states. Finally, we show that Zak states do not naturally exhibit non-trivial topological phases, but by reducing the quantum parameter space, some \virg{spurious} topological phases arise.

 \medskip
 \noindent
{\bf Structure of the paper.}
{\bf Section~\ref{sec:obs_alg_sts2}}
provides an exposition of the setting on which we study the equivalence of states, namely the Weyl $C^*$-algebra, its main subalgebra of observables and its $C^*$-bundle structure, and consequences of this structure on the relevant families of states. 
{\bf Section~\ref{sec:classif_sts_tras2}} deals with the topology of configurations of states that are invariant under all translations, showing that their associated topological phases are trivial in the spin-less case, and that the expected non-trivial phases arise under the addition of spin. 
{\bf Section~\ref{sec:classif_sts_latt2}} is concerned with the topological classification of configurations of semi-regular states invariant under discrete translations. Theorem \ref{theo_main_top} shows that our classification scheme recovers the $K$-theoretic classification of gapped spectral projectors for TIs of type A in the class of \emph{Bloch-wave states}. 
In {\bf Section \ref{app-zac2}} we briefly study Zak states as an additional class of interesting states, also showing their topological triviality.
Finally, in {\bf Section \ref{sec:degeneracy}}, we consider topological phases obtained by allowing sample spaces of non-pure states, complementing the results of Section~\ref{sec:classif_sts_latt2}. Theorem \ref{theo_main_top-2} again recovers the classification of gapped spectral projectors for TIs of type A, now in slightly greater generality, and \ref{theo_main_top-3} does the same for TIs of type AI. We then proceed to show that thermal states present no non-trivial topological phases, as is expected.

\medskip

 \noindent
{\bf Acknowledgements.} SR's research has been supported by the European Research Council under the European Union's Horizon Europe research and innovation programme (ERC grant AMEN-101124789).
GD's research is supported by the grant \emph{Fondecyt Regular - 1230032}. 
{This research was partially supported by the University of Warsaw Thematic Research Programme ``Quantum Symmetries".}

\section{Observable algebras and sample state spaces}
\label{sec:obs_alg_sts2}

In this section, we will present the abstract construction for CCR $C^*$-algebras, which are generalizations of the Weyl $C^*$-algebra. Our exposition will be mainly based on the papers \cite{manuceau-sirugue-testard-verbeure-73,beaume-manuceau-pellet-sirugue-74}. The Weyl $C^*$-algebra  represents the smallest $C^*$-algebra containing the \emph{canonical commutation relations} (CCR) with respect to the phase space $\R^d\times\R^d$ endowed  with its usual symplectic structure. The abstract construction sketched here  generalized the usual construction to generic phase spaces. This fact is relevant to construct relevant $C^*$-subalgebras of the Weyl $C^*$-algebra $\bb{W}$.

\medskip

\subsection{Abstract CCR \texorpdfstring{$C^*$}{TEXT}-algebras}
Following \cite{beaume-manuceau-pellet-sirugue-74}, we introduce here the $C^*$-algebra encoding the general form of the CCR. 
Let $(\n{A},+)$ be an abelian group with neutral element $0$. Let us denote by $\n{S}^1$ the group of complex numbers of unit modulus. A \emph{bicharacter} is a function $\varpi:\n{A}\times\n{A} \to \n{S}^1$ such that 
\begin{align*}
    \varpi (z_1,z_2) \;&=\; \overline{\varpi(z_2,z_1)} \;,\\
    \varpi(z_1,z_2+z_3) \;&=\; \varpi (z_1,z_2) \varpi (z_1,z_3) \;,
\end{align*}
for all $z_1,z_2,z_3\in \n{A}$.
We will further assume the normalization $\varpi(z,z) = 1$ for each $z\in\n{A}$. 
The subgroup
\[
\n{A}_0 \;:=\; \left\{ z\in \n{A} \;\left|\; \varpi(z,z')^2 = 1\; ,\;\;\forall\; z'\in \n{A} \right\}\right.
\]
will be called the \emph{degeneracy group} of $(\n{A},\varpi)$. The bicharacter $\varpi$ is said to be non-degenerate if $\n{A}_0 = \{0\}$.

\medskip

Let $\Delta(\n{A},\varpi)$ be the $*$-algebra generated by non-zero elements $w_z$ with $z\in \n{A}$ satisfying the product law
\begin{equation}\label{eq:abs_ccr_rels_joint}
w_z w_{z'} \;=\; \varpi(z,z')\; w_{z+z'} \;,
\end{equation}
and endowed  with the $*$-involution
\[
w_z^* = w_{-z} \;.
\]
for every $z,z'\in\n{A}$.
We then take $\overline{\Delta(\n{A},\varpi)}$ to be its completion with respect to the minimal regular norm, as in \cite[Section 3]{denittis-rendel-25}, \cite{manuceau-sirugue-testard-verbeure-73}, or \cite{beaume-manuceau-pellet-sirugue-74}. The resulting $C^*$-algebra is simple if and only if $\n{A}_0 = \{0\}$ \cite[Corollary 4.24]{manuceau-sirugue-testard-verbeure-73}, and separable if and only if $\n{A}$ is countable \cite[Corollary 4.25]{manuceau-sirugue-testard-verbeure-73}.

\medskip

We will be mainly interested in the case in which $\n{A} = \n{G}\times \n{G}'$ with $\n{G}$ and $\n{G}'$ both abelian groups such that there exists a pairing $\langle\cdot,\cdot\rangle : \n{G}\times\n{G}' \to \R$ satisfying: i) for each $g\in \n{G}$ the map $\expo{\ii\langle g,\cdot\rangle}: \n{G}'\to \n{S}^1$ is a character of $\n{G}'$; ii) for each $g'\in \n{G}'$, the map $\expo{\ii\langle\cdot,g'\rangle}: \n{G}\to \n{S}^1$ is a character of $\n{G}$. In such a case we assume the map $\varpi$ is defined by
\begin{equation}\label{eq:varpi_dec_def}
    \varpi\big((g_1,g'_1), (g_2,g'_2) \big) \;:=\; \expo{  \frac{\ii}{2} \big(\langle g_1,g'_2\rangle - \langle g_2,g'_1\rangle  \big) } \;.
\end{equation}
In this case, there is an alternative presentation for the $C^*$-algebra $\overline{\Delta(\n{G}\times \n{G}',\varpi)}$. In fact, the latter is generated by  unitary elements $u_g := w_{(g,0)}$ and $v_{g'}:=w_{(0,h)}$, where $g\in \n{G}$ and $g'\in \n{G}'$. The product laws are
\begin{equation}\label{eq:abs_ccr_rels_separate}
    u_g v_{g'} \;=\; \expo{\ii \langle g,g'\rangle }v_{g'} u_g \;,\quad u_{g} u_{h} \; =\; u_{g + h}\;,\quad     v_{g'} v_{h'}\; =\; v_{g' + h'}\;
\end{equation}
and the $\ast$-involution acts as $u_g ^*=u_{-g}$, $v_{g'}^*=v_{-g'}$
The inverse relation on the generators is given by $w_{(g,g')} = \expo{-\frac{\ii}{2} \langle g,g'\rangle} u_g v_{g'}$. 
A typical  example of this construction is provided by $\n{G}$  a real topological vector space, $\n{G}'$   its topological dual, and $\langle\cdot,\cdot\rangle$ the evaluation pairing.

\subsection{Weyl \texorpdfstring{$C^*$}{TEXT}-algebra and invariant subalgebras}\label{sec:field_obs_loc_algs}
The ($d$-dimensional) \emph{Weyl $C^*$-algebra} is the CCR $C^*$-algebra defined by the choices 
$\n{G}=\n{G}'=\R^d$ with $d\in \N$ and $\langle\alpha,\beta \rangle  := \alpha\cdot\beta$ the usual scalar product for every $\alpha,\beta\in\R^d$. After fixing $\varpi$  by \eqref{eq:varpi_dec_def}, we will denote by
\begin{equation}\label{eq_W_pres}
\bb{W}\; :=\; \overline{\Delta(\R^{d}\times\R^d,\varpi)}
\end{equation}
the related (Weyl) $C^*$-algebra. It is known that $\bb{W}$ is a simple, non-separable $C^*$-algebra,  as a consequence of the degeneracy group being trivial and $\R^d$ being uncountable.

\medskip

Given $\lambda\in\R^d$, we define the \emph{spatial translation by $\lambda$} as the $\ast$-automorphism $\tau_\lambda$ of $\bb{W}$ defined by $\tau_\lambda(a) := v_\lambda a v_\lambda^*$, for $a\in\bb{W}$. The assignation $\lambda\mapsto \tau_\lambda$ is strongly continuous, so spatial translations correspond to a strongly continuous action of $\R^d$ on $\bb{W}$. 

\medskip

The \emph{fully-invariant} $C^*$-subalgebra $\bb{V}$ is made up of all elements of $\bb{W}$ invariant under all spatial translations. In other words, $\bb{V}$ corresponds to the fixed-point subalgebra of this $\R^d$-action. One has that (see \cite{denittis-rendel-25}, section 3.2)
\[
\bb{V} \;:=\; \left\{ a\in \bb{W} \;|\; \tau_\lambda(a) = a \;,\; \forall\; \lambda\in\R^d \right\} 
\;=\;
 C^*\!\left(v_\beta\;|\; \beta\in\R^d\right)
\]
This is clearly an abelian $C^*$-algebra, and from the characterization above one can deduce that
\[
\bb{V} \;\simeq\; C^*(\R^d_{\rm d}) \;\simeq\; C(\rr{b}(\R^d)) \;,
\]
where $\R^d_{\rm d}$ denotes the group $\R^d$ with the discrete topology, and $\rr{b}(\R^d)$ the Bohr compactification of $\R^d$.
Viewing $\bb{W}$ according to the presentation \eqref{eq_W_pres}, one infers that
\[
\bb{V} \;=\; \overline{\Delta(\{0\}\times\R^d,\varpi_\tau)}\;.
\]
with $\varpi_\tau := \varpi|_{\{0\}\times\R^d}$. Note that now $\varpi_\tau$ is degenerate, with degeneracy group
\[
(\{0\}\times\R^d)_0 \;=\; \{0\}\times\R^d \;.
\]

\medskip

Let $\Gamma\subset\R^d$ be a maximal lattice (discrete subgroup of maximal rank), \ie $\Gamma\simeq\Z^d$.
The \emph{$\Gamma$-invariant} $C^*$-subalgebra $\bb{V}_\Gamma\subset \bb{W}$ is composed by the elements that are invariant under the action of the subgroup $\Gamma\subset \R^d$ by spatial translations. 
Then one has that
\[
\bb{V}_\Gamma \;:=\; \left\{ a\in \bb{W} \;|\; \tau_\gamma(a) = a \;,\; \forall\; \gamma\in\Gamma \right\} 
\;=\;
 C^*\!\left(u_{\gamma'} v_\beta\;|\; \gamma'\in \Gamma', \beta\in\R^d\right) 
\]
where $\Gamma'$ is the \emph{dual lattice} of $\Gamma$. The latter is defined as the set of $\gamma'\in\R^d$ such that $\gamma\cdot\gamma'\in 2\pi\Z$ for every $\gamma\in\Gamma$. 
It turns out that $\Gamma'\simeq\Z^d$ is again a maximal lattice.
Again, viewing $\bb{W}$ according to its symplectic presentation \eqref{eq_W_pres}, one gets
\[
\bb{V}_\Gamma \;=\; \overline{\Delta(\Gamma'\times\R^d,\varpi_\Gamma)}\;.
\]
with $\varpi_\Gamma := \varpi|_{\Gamma'\times\R^d}$. Once again $\varpi_\Gamma$ is degenerate  with degeneracy group 
\[
(\Gamma'\times\R^d)_0 \;=\; \{0\}\times\Gamma \;.
\]
The center of $\bb{V}_\Gamma$ is given by
\[
\bb{Z}(\bb{V}_\Gamma) \;=\; C^*(v_\gamma \;|\;\gamma\in\Gamma) \;=\;  \overline{\Delta(\{0\}\times\Gamma,1)} \;\simeq\; C^*(\Gamma) \;\simeq\; C(\n{B}_\Gamma) \;
\]
where $\n{B}_\Gamma:=\R^d/\Gamma'\simeq\n{T}^d$ is called the \emph{Brillouin torus}, and we canonically identify it with $\widehat{\Gamma}$, the Pontryagin dual of ${\Gamma}$.
The first equality is an application of \cite[Theorem 4.2]{manuceau-sirugue-testard-verbeure-73}. The second one is a direct check. The first isomorphism follows by a natural  identification of generators. The last one is the Gelfand isomorphism.

\medskip

Just as done with spatial translations, we can define \emph{momentum translations}: Given $\eta\in\R^d$, the corresponding momentum translation is given by $\theta_\eta(a) := u_\eta a u_\eta^*$. 

\medskip

We may also define the joint action of $\Gamma\times\Gamma'$ on $\bb{W}$ as $(\gamma,\gamma')\mapsto \zeta_{\gamma,\gamma'}:=\theta_{\gamma'}\circ\tau_\gamma$. This action is well-defined as a group homomorphism since $u_{\gamma'}$ commutes with $v_\gamma$ for all $(\gamma,\gamma')\in\Gamma\times\Gamma'$. 
The subalgebra of elements of $\bb{W}$ invariant under this action is called the \emph{Zak algebra}, denoted $\bb{Z}_\Gamma$. It is characterized as follows:
\[
\bb{Z}_\Gamma \;:=\; \left\{ a\in \bb{W} \;|\; \zeta_{\gamma,\gamma'}(a) = a \;,\; \forall\; (\gamma,\gamma')\in\Gamma\times\Gamma' \right\} 
\;=\;
 C^*\!\left(u_{\gamma'} v_\gamma\;|\; (\gamma,\gamma')\in\Gamma\times\Gamma'\right) \;.
\]
Again, this is an abelian $C^*$-subalgebra of $\bb{W}$. Its symplectic form is obtained by restricting $\varpi$ to $\Gamma'\times\Gamma$, which again coincides with the associated degeneracy group.

\medskip

Finally, let us introduce an
additional $C^*$-algebra, which will correspond to a \virg{local algebra} of the system in the case of $\Gamma$-translations, in a sense that will be clarified in the following. 
Given the \emph{unit cell}
$\n{T}_\Gamma:=\R^d/\Gamma$, which
can be again identified with a $d$-dimensional torus,
let us define the map $\langle\cdot,\cdot\rangle_{\rm loc}:\Gamma'\times\n{T}_\Gamma \to \R$ via
\[
\langle\gamma', [\beta]\rangle_{\rm loc} \;:=\; \gamma'\cdot y_\beta \;, \qquad \forall\; \gamma'\in \Gamma'\;,\;\;\ [\beta]\in \n{T}_\Gamma
\]
where $y_\beta\in\R^d$ is the representative of $[\beta]$ in the fundamental cell $Q_\Gamma$ of $\Gamma$, and $\gamma_\beta:=\beta-y_\beta\in\Gamma$ is the lattice remainder.
Note that 
\[
\expo{\ii \langle\gamma', [\beta]\rangle_{\rm loc}} \;:=\; \expo{\ii \gamma'\cdot y_\beta} \;=\; \expo{\ii (\gamma'\cdot \beta-\gamma'\cdot\gamma_\beta)} \;=\; \expo{\ii \gamma'\cdot \beta}  \;.
\]
So, fixing any of the two entries in $\expo{\ii \langle\cdot, \cdot\rangle_{\rm loc}}$, one obtains  a character of the respective group. Then, one may define $\varpi_{\rm loc}$ as done in \eqref{eq:varpi_dec_def}, and introduce
\begin{equation}\label{eq:A_Gamma_def}
    \bb{A}_\Gamma \;:=\; \overline{\Delta(\Gamma'\times\n{T}_\Gamma,\varpi_{\rm loc})}\;.
\end{equation}
The elements of $\bb{A}_\Gamma $ can be interpreted as $\Gamma$-invariant observables  restricted to the unit cell  $\n{T}_\Gamma$.

\section{Topological phases  of translationally invariant pure states}\label{sec:classif_sts_tras2}

In \cite[Section 4]{denittis-rendel-25} we studied the topology of translationally invariant pure states of $\bb{W}$. In this section, we will make use of those results to provide some information about the classification of the related topological phases.
We will follow the program described in Section \ref{sec:Intr0} as roadmap. The first step is to declare the symmetry group. In this case it will be $\R^d$ acting as the group of translations. Its action is still be denoted by $\lambda\mapsto\tau_\lambda$.

\subsection{The sample space of fully-invariant states}
The next step amounts to the selection of the relevant sample space of states. We will choose the set $\s{P}_\bb{W}^\tau$ of all pure, translation-invariant states of $\bb{W}$.  
In order to make the distinction with states that are translationally invariant under the subgroup of translations by a lattice $\Gamma\subset\R^d$, we will often refer to these states as \emph{fully-invariant} states.

\subsection{\texorpdfstring{$C^*$}{C*}-bundle structure of the fully-invariant subalgebra}\label{sec_V-fb}
We proceed to the next step. It is immediate that any abelian $C^*$-algebra is a trivial $C^*$-bundle over its spectrum. 
From Section \ref{sec:field_obs_loc_algs}
One has that 
\[
\bb{V} \;\simeq\; C(\rr{b}(\R^d)).
\]
Therefore, $\bb{V}$ is a trivial $C^*$-bundle over $\rr{b}(\R^d)$ with typical fiber $\C$.
Any pure state of $\bb{V}$ can be identified with the evaluation at a point of $\rr{b}(\R^d)$.
Therefore, the space of pure states $\s{P}_\bb{V}$ can be seen as the total space of a trivial fiber bundle $\widetilde{\pi}_{\tau}:\s{P}_\bb{V} \to \rr{b}(\R^d)$ where $\widetilde{\pi}_{\tau}$ is the natural homeomorphism induced by the Gelfand duality. The typical fiber of the bundle is a singleton. 


\begin{remark}[Interpreting the Bohr compactification]
Let us recall that $\R^d$ as an abelian group coincides with its Pontryagin dual (or spectrum). As such, the Bohr compactification of $\R^d$ appears as a possible compactification of the dual or spectrum of the symmetry group $\R^d$. The choice of compactification can be interpreted in terms of \virg{boundary conditions at infinity} imposed on the observables, seen formally as functions of this quantum number. 
Using the Bohr compactification amounts to taking observables of the preserved quantum number (the momentum) that \virg{oscillate to infinity}, in the sense of being almost periodic. 
In the case of the Weyl $C^*$-algebra, the elements invariant under translations are those generated by the elements $v_\beta$, which can be thought of as exponentials of the momentum. This fact is formalized by the $*$-isomorphism $\bb{V}\simeq {\rm AP}(\R^d)$ given by $v_\beta\mapsto \expo{\ii\beta\cdot(\cdot)}$ which identifies $\bb{V}$ with the almost periodic functions of the momentum variable. As a consequence not just any function of momentum is allowed in $\bb{V}$. For example, the only functions   $f\in{\rm AP}(\R^d)$ that admit a limit
\[
\lim_{|p|\to\infty}f(p)\;=\;L
\]
are the constants, \ie $f(p)=L$ for every $p\in\R^d$. 
Therefore, we cannot think of $\bb{V}$ as a $C^*$-bundle over $\R^d$, or over its one-point compactification $\n{S}^d$ in this manner. 
\hfill $\blacktriangleleft$
\end{remark}

\subsection{Translationally invariant states as a fiber bundle}
The identification $\s{P}_\bb{W}^\tau\simeq \s{P}_\bb{V}$ and the discussion in Section \ref{sec_V-fb}
show  that the sample space $\s{P}_\bb{W}^\tau$ is the total space of a trivial fiber bundle. Indeed, if $\iota_\tau$ denotes the homeomorphism $\s{P}_\bb{W}^\tau\rightarrow \s{P}_\bb{V}$ given by $\iota_\tau(\omega) = \omega|_\bb{V}$, then the bundle map is given by $\pi_\tau:=\widetilde{\pi}_{\tau}\circ\iota_\tau: \s{P}_\bb{W}^\tau \to \rr{b}(\R^d)$. This bundle map corresponds to the homeomorphism $\pi_{\tau}:\omega_\lambda\mapsto\lambda$ presented in \cite[Proposition 4.1]{denittis-rendel-25}.

\subsection{Topological phases associated to fully-invariant states}
We are moving towards the final steps.
Since the bundle map $\pi_\tau$ is an homeomorphism it follows that $\pi_\tau$ has a unique continuous section corresponding to its inverse. Namely,
\[
{\rm Sec}(\rr{b}(\R^d),\s{P}_\bb{V})\;=\;\{\pi_\tau^{-1}\}
\]
is a singleton. In view of this one ends with
\[
\Omega(\s{P}_\bb{W}^\tau) \;:=\; \big[{\rm Sec}(\rr{b}(\R^d),\s{P}_\bb{V})\big] \;=\;  \{\ast\}
\]
meaning that there is a unique (trivial)
 topological phase.

\subsection{Adding spin}
In view of the above, the fully-invariant states of $\bb{W}$ carry no topological information. However, this wouldn't be the case if one were to add \emph{spinorial} degrees of freedom.
Let $\bb{W}_{\rm spin}:=\bb{W}\otimes \bb{S}$ the tensor product of the Weyl algebra (which is nuclear) with a \emph{spin  algebra} $\bb{S}$. In concrete examples we can think of $\bb{S}$ as a matrix algebra or, more in general, as the algebra of the compact operators.  
In such a case the fully invariant subalgebra has the form $\bb{V}\otimes \bb{S} \simeq C(\rr{b}(\R^d))\otimes\bb{S}$ and, as a consequence of \cite[Proposition 2.9]{denittis-25}, 
one gets that $\s{P}_{\bb{W}_{\rm spin}}^\tau\simeq\s{P}_{\bb{V}\otimes\bb{S}}\simeq \rr{b}(\R^d)\times \s{P}_\bb{S}$. It turns out that $\s{P}_{\bb{W}_{\rm spin}}^\tau$ can be identified with the total space of the trivial fiber bundle  $\pi_{\tau}:\rr{b}(\R^d)\times \s{P}_\bb{S}\to \rr{b}(\R^d)$. The continuous sections of this bundle are in bijection with continuous functions $\rr{b}(\R^d) \to \s{P}_\bb{S}$ and one ends with 
\[
\Omega(\s{P}_{\bb{W}_{\rm spin}}^\tau)\;=\;[\rr{b}(\R^d),\s{P}_\bb{S}]\;\simeq\;[\rr{b}(\R^d),B\n{U}(1)]
\] 
as the set of topological phases. In the last bijection $B\n{U}(1)$ is the classification space for $\n{U}(1)$ principal bundles and the result is justified as in  \cite[Theorem 1.7]{denittis-25} assuming that 
$\bb{S}$ is (a sub algebra of) the algebra of the compact operators.

\begin{remark}[$K$-theory of the   Bohr compactification]
We got the result
 that $\Omega(\s{P}_{\bb{W}_{\rm spin}}^\tau)$ is classified by isomorphism classes of line bundles over $\rr{b}(\R^d)$. Classifying these maps would be no easy task, given the pathological properties of the Bohr compactification. A simpler, related, undertaking might be  to compute the (reduced) topological $K$-theory of $\rr{b}(\R^d)$, which could yield interesting information regarding the (stable) isomorphism classes of vector bundles over it. Since $\rr{b}(\R^d)$ has uncountably infinite path-connected components, the $K$-theory is likely to be quite big. 
 Moreover, since $\rr{b}(\R^d)$ is also connected, it is reasonable to suspect that its topology is simply the sum of its path-connected  components
To the best of our knowledge, we have not found any results in the literature on K-theory of the Bohr compactification to date. A related problem is the description of the $K_0$ group of the Weyl $C^*$-algebra $\bb{W}\simeq C(\rr{b}(\R^d))\rtimes\R^d_{\rm{d}}$.
\hfill $\blacktriangleleft$
\end{remark} 


\subsection{The regular case}

We would now like to restrict our attention to the $\beta$-regular translationally invariant states. In the literature, these are said to be parametrized by $\R^d$, being exactly the states of the form $\omega_{\rr{i}(p)}$, $p\in\R^d$. Here, $\rr{i}:\R^d\to \rr{b}(\R^d)$, is the Bohr compactification map, which is a continuous and injective group homomorphism whose inverse is not continuous. As such, we have an homeomorphism $\s{P}_\bb{W}^{\tau,\beta} \simeq \R^d_\rr{b}:=\rr{i}(\R^d)$ \cite[Proposition 4.1]{denittis-rendel-25}, corresponding to the restriction of $\s{P}_\bb{W}^{\tau} \simeq \rr{b}(\R^d)$. Thus, if we wanted to consider $\s{P}_\bb{W}^{\tau,\beta}$ as a fiber bundle of states as above, we would need to deal with the poor properties of $\R^d_\rr{b}$, such as having a strictly weaker topology than the relevant space $\R^d$ and being non-locally compact. Luckily, there's a natural way to deal with this issue in general.

\medskip

Suppose $\n{G}$ is a LCA group and $X$ is a compactification of the dual group $\widehat{\n{G}}$ in the sense that there is a continuous injective function $\rr{i}:\widehat{\n{G}}\to X$ with dense image. Let $\pi:E\to X$ be a fiber bundle over $X$. Then, we can consider the fiber bundle $\rr{i}^*(\pi):\rr{i}^*E\to \widehat{\n{G}}$ induced by $\rr{i}$. Here, $\rr{i}^*E = \{(\omega,p)\in E\times \widehat{\n{G}} \;|\; \pi(\omega) = p\}$, and $\rr{i}^*(\pi)(\omega,p) = \pi(\omega) = p$. If $\rr{i}$ is an homeomorphism onto its image, then this induced fiber bundle is isomorphic to the restricted bundle $\pi: \pi^{-1}(\widehat{\n{G}})\to \rr{i}(\widehat{\n{G}})\simeq\widehat{\n{G}}$, but in general, the base space of this restricted bundle will have a weaker topology. The induced bundle construction allows us to recover the topological information of $\widehat{\n{G}}$ that was \virg{lost} upon compactifying it.

\medskip

Applying this on our current setting, we get the following: the total space of the induced bundle is $\rr{i}^*(\s{P}_\bb{W}^{\tau}) = \{(\omega_p,p)  \;|\; p\in \R^d\} \subset \s{P}_\bb{W}^{\tau,\beta}\times\R^d$, which is in bijection with $\s{P}_\bb{W}^{\tau,\beta}$, and the bundle map is $\rr{i}^*(\pi_\tau)(\omega_p,p) = \pi_\tau(\omega_p) = p$. Note that $\rr{i}^*(\pi_\tau):\rr{i}^*(\s{P}_\bb{W}^{\tau}) \to \R^d$ is now an homeomorphism with inverse $\pi_\tau^{-1}(p) = (\omega_p,p)$. Under this framework, the topological phases remain trivial in the spin-less case:
\[
\Omega(\s{P}_\bb{W}^{\tau,\beta}) \;:=\; \big[{\rm Sec}(\R^d,\rr{i}^*(\s{P}_\bb{W}^{\tau}))\big] \;=\; \{[\pi_\tau^{-1}] \} \;=\; \{\ast\}\;.
\]

\medskip

Similarly to before, one can add spin by considering $\bb{W}_{\rm{spin}}:= \bb{W}\otimes \bb S$. 
Note that $\s{P}_\bb{W_{\rm spin}}^{\tau} \simeq \s{P}_{\bb{V}\otimes \bb{S}} \simeq \rr{b}(\R^d) \times \s{P}_\bb{S}$, and under this identification, $\rr{i}^*(\s{P}_\bb{W_{\rm spin}}^{\tau}) \simeq \R^d \times \s{P}_\bb{S}$, which is to say that $\pi_\tau$ and $\rr{i}^*(\pi_\tau)$ are again trivial.
Due to the triviality of $\rr{i}^*(\pi_\tau)$, one again has the homotopy-preserving bijection
\[
{\rm Sec}(\R^d,\rr{i}^*(\s{P}_\bb{W_{\rm spin}}^{\tau})) \;\simeq\; C (\R^d, \s{P}_\bb{S}) \;.
\]
Explicitly, this is given by the following: Due to the triviality of the bundle, sections $F\in{\rm Sec}(\R^d,\rr{i}^*(\s{P}_\bb{W_{\rm spin}}^{\tau}))$ can be conceived as functions of the form $F(p) = (p,f(p))$, with $f\in C(\R^d, \s{P}_\bb{S})$. Then, two sections $F$ and $ F'$ are homotopic if and only if their respective second coordinates $f$ and $f'$ are, which gives rise to the correspondence.

\medskip

\begin{remark}[Strong topological invariants]
In the standard description of strong topological invariants for topological insulators, one considers homotopy classes of Hamiltonians or states parametrized by $\R^d$, having a definite limit at infinity. To emulate this, we can consider the sections $F = (\cdot,f(\cdot))\in {\rm Sec}(\R^d,\rr{i}^*(\s{P}_\bb{W_{\rm spin}}^{\tau}))$, such that $\lim_{|p|\to\infty} f(p)$ exists. 
We denote the set of such sections by ${\rm Sec}_\infty(\R^d,\rr{i}^*(\s{P}_\bb{W_{\rm spin}}^{\tau}))$. Then, restricting the above correspondence, we have
\[
{\rm Sec}_\infty(\R^d,\rr{i}^*(\s{P}_\bb{W_{\rm spin}}^{\tau})) \;\simeq\; C_\infty (\R^d, \s{P}_\bb{S}) \;\simeq\; C(\n{S}^d,\s{P}_\bb{S}) \;.
\]
From this, the corresponding topological phases are
\[
\Omega_{\infty}(\s{P}_{\bb{W}_{\rm spin}}^\tau) \;:=\; \big[{\rm Sec}_\infty(\R^d,\rr{i}^*(\s{P}_\bb{W_{\rm spin}}^{\tau}))\big] \;\simeq\; [\n{S}^d,\s{P}_\bb{S}] \;.
\] 

\medskip

Thus, in the case $\bb{S} = \bb{K}$, where $\bb{K}$ denotes the compact operators on an infinite-dimensional, separable Hilbert space, one gets
\[
\Omega_{\infty}(\s{P}_{\bb{W}_{\rm spin}}^\tau) \;\simeq\;[\n{S}^d,\s{P}_\bb{S}] \;\simeq\;[\n{S}^d,B\n{U}(1)] \;\simeq\; \rm{Vec}^1_\C(\n{S}^d) \;,
\] 
where $\rm{Vec}^1_\C(\n{S}^d)$ denotes the isomorphism classes of complex line bundles over the $d$-dimensional sphere. Note that in low dimension $1\leq d\leq3$, this implies 
\[
\Omega_{\infty}(\s{P}_{\bb{W}_{\rm spin}}^\tau) \;\simeq\; \widetilde{K}_0(\n{S}^d) \;.
\]
This recovers the classification of the strong topological invariants for topological insulators of type A.
\hfill $\blacktriangleleft$
\end{remark}

\section{Topology of pure lattice-invariant states}\label{sec:classif_sts_latt2}
In this Section, we will   study topological phases associated to discrete translation symmetries  following step by step the recipe described in Section \ref{sec:Intr0}. 
The first step consists in fixing the symmetry group as a maximal lattice
 $\Gamma\subset\R^d$   as in Section \ref{sec:field_obs_loc_algs}. The action of $\Gamma$ on $\bb{W}$ by translations  will be denoted by  $\Gamma\ni\gamma\mapsto \tau_\gamma\in{\rm Aut}(\bb{W})$. Of course, since $\Gamma$ is abelian discrete, its spectrum (Pontryagin dual) $\n{B}_\Gamma$ is already compact. It will be the natural choice for our quantum parameter space.

\subsection{The sample space of Bloch-waves}

For the sample space of $\Gamma$-invariant state, we will employ the set of Bloch-wave states $\s{P}^{\Gamma,\beta}_\bb{W}$, which consists of the pure, $\Gamma$-invariant, $\beta$-regular states of $\bb{W}$, described in \cite[Proposition 3.26]{denittis-rendel-25}. In particular, it will become evident that the fiber bundle structure of these states described in \cite[Theorem 5.1]{denittis-rendel-25} fits nicely into our classification scheme.

\subsection{\texorpdfstring{$C^*$}{C*}-bundle structure of the lattice-invariant subalgebra}\label{sec:cs_bundle_obs}
Let $\bb{V}_\Gamma\subset \bb{W}$ be the $\Gamma$-invariant $C^*$-subalgebra described in Section \ref{sec:field_obs_loc_algs}. Our next goal is to show that $\bb{V}_\Gamma$ is a $C^*$-bundle over $\n{B}_\Gamma$ and characterize its fibers. We begin by showing that $\bb{V}_\Gamma$ is a $C(\n{B}_\Gamma)$-algebra. 
As a preliminary fact let us note that $\bb{V}_\Gamma$ is unitary and therefore it coincides with its multiplier algebra.
Moreover,  we have an inclusion $C^*(\Gamma)\hookrightarrow\bb{V}_\Gamma$ given on generators by $\delta_{\gamma} \mapsto v_\gamma$, for $\gamma\in \Gamma$. Moreover, since the elements $v_\gamma$ generate the center of $\bb{V}_\Gamma$, this map is a $*$-isomorphism $C^*(\Gamma) \simeq \bb{Z}(\bb{V}_\Gamma)$. Composing with the Gelfand isomorphism $C^*(\Gamma)\simeq C(\n{B}_\Gamma)$, this provides the structure (iso)morphism $\Xi:C(\n{B}_\Gamma) \rightarrow \bb{Z}(\bb{V}_\Gamma)$ mapping the phases $\expo{-\ii\gamma\cdot(\cdot)}$ to the elements $v_\gamma$. 
\medskip

As described in Section \ref{sec:sts_Cs_bundles}, for each $\kappa\in\n{B}_\Gamma$, we define the closed, two-sided ideal 
$$
\bb{I}_\kappa \;:=\; \Xi(C_0(\n{B}_\Gamma\setminus\{\kappa\})) \bb{V}_\Gamma\; =\; \left\{\Xi(g) a \in \bb{V}_\Gamma \;|\; g\in C_0(\n{B}_\Gamma\setminus\{\kappa\}), a\in \bb{V}_\Gamma \right\}\;.
$$

\begin{proposition}\label{prop:ideal_structure_Ik}
The ideal $\bb{I}_\kappa$ is generated by the elements
\[
(v_\gamma-\expo{-\ii \gamma\cdot \kappa} {\bf 1}) a \;, 
\]
with $\gamma\in \Gamma$ and $a\in\bb{V}_\Gamma$. Moreover, for each $\alpha\in\R^d$, the momentum translation $\theta_\alpha$ provides a $*$-isomorphism $\bb{I}_\kappa\to\bb{I}_{\kappa+\alpha}$, and $\theta_\alpha(\Xi(g)) = \Xi(g(\cdot-\alpha))$.
\end{proposition}
\proof
First of all $i_{\gamma,a}^{(\kappa)}:=(v_\gamma-\expo{-\ii \gamma\cdot \kappa} {\bf 1}) a\in \bb{I}_\kappa$ by construction.
Let
\[
g(\kappa')\;:=\;\sum_{\gamma\in\Gamma}c_\gamma \expo{-\ii \gamma\cdot\kappa'}\;\in\;C_0(\n{B}_\Gamma\setminus\{\kappa\})\;.
\]
Then $g(\kappa)=0$ implies $\sum_{\gamma\in\Gamma}c_\gamma \expo{-\ii \gamma\cdot\kappa}=0$.
Therefore, from $g(\kappa')=g(\kappa')-0$, and using the latter expansion for $0$, one gets 
\[
g(\kappa')\;=\;\left(\sum_{\gamma\in\Gamma}c_\gamma \expo{-\ii \gamma\cdot\kappa'}\right)-\left(\sum_{\gamma\in\Gamma}c_\gamma \expo{-\ii \gamma\cdot\kappa}\right)\;=\;\sum_{\gamma\in\Gamma}c_\gamma \left(\expo{-\ii \gamma\cdot\kappa'} -  \expo{-\ii \gamma\cdot\kappa}\right)\;.
\]
With this, and using that $\Xi(\expo{-\ii\gamma\cdot(\cdot)}) = v_\gamma$, a generic element of $\bb{I}_\kappa$ has the form
\[
\Xi(g)a \;=\; \left(\sum_{\gamma\in\Gamma}c_\gamma \left(v_\gamma -  \expo{-\ii \gamma\cdot\kappa} {\bf 1} \right)\right) a \;=\; \sum_{\gamma\in\Gamma}c_\gamma i_{\gamma,a}^{(\kappa)}\;.
\]
This proves that the elements $i_{\gamma,a}^{(\kappa)}$ generate the ideal. Finally, let us observe that
\[
\theta_\alpha \left(v_\gamma -  \expo{-\ii \gamma\cdot\kappa} {\bf 1} \right)  \;=\; \expo{\ii \gamma\cdot\alpha}v_\gamma -  \expo{-\ii \gamma\cdot\kappa} {\bf 1} \;=\; \expo{\ii \gamma\cdot\alpha} \left(v_\gamma -  \expo{-\ii \gamma\cdot(\kappa+\alpha)} {\bf 1}\right)
\]
which implies that 
$$\theta_\alpha(i_{\gamma,a}^{(\kappa)})\;=\;\expo{\ii \gamma\cdot\alpha} i_{\gamma,\theta_\alpha(a)}^{(\kappa+\alpha)}
$$
Since $\theta_\alpha:\bb{V}_\Gamma \to \bb{V}_\Gamma$ is already a $*$-isomorphism 
one conclude that $\theta_\alpha:\bb{I}_\kappa\to \bb{I}_{\kappa+\alpha}$ is surjective. A similar calculation shows that $\theta_{-\alpha}$ is the inverse. Therefore $\theta_\alpha$
provides a $*$-isomorphism $\bb{I}_\kappa\to \bb{I}_{\kappa+\alpha}$. 
The formula $\theta_\alpha(\Xi(g)) = \Xi(g(\cdot-\alpha))$ follows from the fact that we proved it on generators.
\qed


\medskip

We also have the following result, which is an immediate consequence of \cite[Corollary 4.21]{manuceau-sirugue-testard-verbeure-73}.

\begin{proposition}\label{prop-max}
    The maximal ideals of $\bb{V}_\Gamma$ are exactly the ideals $\bb{I}_\kappa$, with $\kappa\in \n{B}_\Gamma$.
\end{proposition}

\medskip

We are now in position to describe the fibers of $\bb{V}_\Gamma$. Let $\rr{h}_\Gamma:=L^2(\n{T}_\Gamma,\dd\nu)$ where $\n{T}_\Gamma:=\R^d/\Gamma$ and $\nu$ its normalized Haar measure. Consider the families of operators given by
\begin{equation}\label{eq_SS-YY2}
\begin{aligned}
(S_\beta f)(y) &\;:=\; f(y-\beta) \\
(F_{\gamma'} f)(y) &\;:=\; \expo{\ii \gamma'\cdot y}f(y)
\end{aligned}
\end{equation}
for every $\gamma'\in\Gamma'$ and $\beta\in\R^d$. These operators and their connection to the irreducible representations of the $C^*$-algebra $\bb{V}_\Gamma$ have been discussed in 
\cite[Section 3.3]{denittis-rendel-25}. 
It holds that the $C^*$-algebra 
\[
\bb{C}(\rr{h}_\Gamma)\;:=\;C^*\left(\left\{F_{\gamma'}S_\beta \;|\; \gamma'\in\Gamma',\beta\in\R^d\right\}\right)\;\subset\; \bb{B}(\rr{h}_\Gamma)
\]
generated by these elements is $*$-isomorphic with the abstract CCR $C^*$-algebra $\bb{A}_\Gamma$ defined in \eqref{eq:A_Gamma_def}.

\begin{lemma}\label{lemma-AG}
    The $C^*$-subalgebra $\bb{C}(\rr{h}_\Gamma)$ is $*$-isomorphic to $\bb{A}_\Gamma$.
\end{lemma}
\proof
Denote by $\tilde{u}_{\gamma'}$ and $\tilde{v}_{[\beta]}$ the generators of $\bb{A}_\Gamma$, with relations given as in \eqref{eq:abs_ccr_rels_separate}. Consider the map $\tilde{u}_{\gamma'} \tilde{v}_{[\beta]} \mapsto F_{\gamma'}S_\beta$, which is well-defined since $\beta\mapsto S_\beta$ is $\Gamma$-periodic. This is clearly a representation of $\bb{A}_\Gamma$ in $\bb{B}(\rr{h}_\Gamma)$, which is injective since $\bb{A}_\Gamma$ is simple. This   implies that the map is a $*$-isomorphism into its image, which contains al the operators 
$F_{\gamma'}S_\beta$. The later generate $\bb{C}(\rr{h}_\Gamma)$, then by continuity 
the map is also 
surjective, and thus a $*$-isomorphism between $\bb{A}_\Gamma$ and $\bb{C}(\rr{h}_\Gamma)$.
\qed

\medskip

In view of the previous result, we will henceforth tacitly identify $\bb{A}_\Gamma$ and $\bb{C}(\rr{h}_\Gamma)$.
For the next result we need to consider the family of representations $\rho_\kappa:\bb{V}_\Gamma\to\bb{B}(\rr{h}_\Gamma)$ described in \cite[Section 3.3]{denittis-rendel-25},  given by
\begin{equation}\label{eq:rep_rho_k}
\begin{aligned}
{\rho}_\kappa (v_\beta)  &\;:=\; \expo{-\ii\kappa\cdot \beta} S_\beta \;,\qquad
{\rho}_\kappa(u_{\gamma'})  \;:=\; F_{\gamma'}\,,
\end{aligned}
\end{equation}
where  $\kappa\in\n{B}_\Gamma$ is identified with its representative in the fundamental cell  of the dual lattice $Q_{\Gamma'}$.

\begin{proposition}\label{prop:typ-fib}
    The fiber $C^*$-algebras $\bb{A}_\kappa:= \bb{V}_\Gamma/\bb{I}_\kappa$ are all $*$-isomorphic to $\bb{A}_\Gamma$. In particular, they are simple, nuclear, non-separable $C^*$-algebras.
\end{proposition}
\proof
Let us show that $\ker (\rho_\kappa) = \bb{I}_\kappa$. 
The fact that $\bb{I}_\kappa \subseteq \ker(\rho_\kappa)$ is immediate from Proposition \ref{prop:ideal_structure_Ik}. In fact, for any $\gamma \in \Gamma$ and $a\in \bb{V}_\Gamma$,
\begin{align*}
    \rho_\kappa \left( (v_\gamma-\expo{-\ii \gamma\cdot \kappa} {\bf 1}) a \right) \;&=\; (\expo{-\ii \gamma\cdot \kappa} S_\gamma-\expo{-\ii \gamma\cdot \kappa} {\bf 1}_{\rr{h}_\Gamma}) \rho_\kappa(a) \\
    &=\; (\expo{-\ii \gamma\cdot \kappa} -\expo{-\ii \gamma\cdot \kappa} ) \rho_\kappa(a) \;=\;0 \;,
\end{align*}
where we used that $S_\gamma = {\bf 1}_{\rr{h}_\Gamma}$ for any $\gamma\in \Gamma$. On the other hand, since $\rho_\kappa({\bf 1}) = {\bf 1}_{\rr{h}_\Gamma} \neq 0$, $\ker(\rho_\kappa)$ is a proper ideal. Then, given that $\bb{I}_\kappa$ is a maximal ideal in view of Proposition \ref{prop-max}
, and $\bb{I}_\kappa \subseteq \ker(\rho_\kappa)$, they must be equal. 
Therefore $\bb{A}_\kappa= \bb{V}_\Gamma/\ker(\rho_\kappa)$.
The first isomorphism theorem of quotient space implies that $\widetilde{\rho}_\kappa:\bb{A}_\kappa\to {\rho}_\kappa(\bb{V}_\Gamma)$ is a $\ast$-isomorphism.
To conclude the argument let us observe that ${\rho}_\kappa(\bb{V}_\Gamma)=\bb{C}(\rr{h}_\Gamma)$ and $\bb{C}(\rr{h}_\Gamma)\simeq \bb{A}_\Gamma$ in view of Lemma \ref{lemma-AG}.
\qed

\medskip

Following Section \ref{sec:sts_Cs_bundles}, let $\pi_\kappa:\bb{V}_\Gamma\to \bb{A}_\kappa$, with $\kappa\in\n{B}_\Gamma$ be the quotient map.

\begin{lemma}\label{lem:cont_bund_str}
The map $\kappa\mapsto \|\pi_\kappa(a)\|_{\bb{A}_\kappa}$ is continuous for every $a\in\bb{V}_\Gamma$.
\end{lemma}
\proof

Since a generic element of $\bb{I}_\kappa$ has the form $\Xi(g) b$ for some $g\in C_0(\n{B}_\Gamma\setminus\{\kappa\})$ and $b\in\bb{V}_\Gamma$, one has that
\begin{align*}
\|\pi_\kappa(a)\|_{\bb{A}_\kappa} \;&=\; \inf \left\{ \|a-\Xi(g)b\| \;\big|\; {g\in C_0(\n{B}_\Gamma\setminus \{\kappa\}),\, b\in\bb{V}_\Gamma} \right\} \\
&=\; \inf \left\{ \| \theta_\alpha (a-\Xi(g)b)\| \;\big|\; {g\in C_0(\n{B}_\Gamma\setminus \{\kappa\}),\, b\in\bb{V}_\Gamma} \right\} \\
&=\;\inf \left\{ \|\theta_\alpha(a)-\Xi(g(\cdot-\alpha))\theta_\alpha(b)\| \;\big|\; {g\in C_0(\n{B}_\Gamma\setminus \{\kappa\}),\, b\in\bb{V}_\Gamma} \right\}\\
&=\;\inf \left\{ \|\theta_\alpha(a)-\Xi(g)b\| \;\big|\; {g\in C_0(\n{B}_\Gamma\setminus \{\kappa+\alpha\}),\, b\in\bb{V}_\Gamma} \right\}\\
&=\; \|\pi_{\kappa+\alpha}(\theta_\alpha(a))\|_{\bb{A}_{\kappa+\alpha}} \;.
\end{align*}
Then, substituting $\alpha = \kappa_0-\kappa$ in the above yields
\[
\lim_{\kappa\to\kappa_0} \|\pi_\kappa(a)\|_{\bb{A}_\kappa} \;=\;\lim_{\kappa\to\kappa_0} \|\pi_{\kappa_0}(\theta_{\kappa_0-\kappa}(a))\|_{\bb{A}_{\kappa_0}} \;=\; \|\pi_{\kappa_0}(a)\|_{\bb{A}_{\kappa_0}} \;,
\]
where the last equation is obtained by noting that
\begin{align*}
\left| \|\pi_{\kappa_0}(\theta_{\kappa_0-\kappa}(a))\|_{\bb{A}_{\kappa_0}} - \|\pi_{\kappa_0}(a)\|_{\bb{A}_{\kappa_0}} \right|  \;&\leq\; \|\pi_{\kappa_0}(\theta_{\kappa_0-\kappa}(a)-a)\|_{\bb{A}_{\kappa_0}} \\
&\leq\; \|\theta_{\kappa_0-\kappa}(a)-a\| \xrightarrow{\kappa\to\kappa_0}0\;.
\end{align*}
in view of the fact that $\alpha\mapsto \theta_\alpha$ is strongly continuous.
Hence the map $\kappa\mapsto \|\pi_\kappa(a)\|_{\bb{A}_\kappa}$ is continuous.
\qed

\medskip

The following result provides a complete description of  $\bb{V}_\Gamma$  as $C^*$-bundle.

\begin{theorem}\label{prop:Vgammabundle}
   $\bb{V}_\Gamma$ has the structure of a locally trivial $C^*$-bundle over $\n{B}_\Gamma$ with typical fiber $\bb{A}_\Gamma$.
   \end{theorem}
\proof
Proposition \ref{prop:typ-fib} and Lemma \ref{lem:cont_bund_str} imply that $\bb{V}_\Gamma$ has the structure of a $C^*$-bundle over $\n{B}_\Gamma$ with typical fiber $\bb{A}_\Gamma$.
Let us now prove its local triviality. We will show next that there exists a compact neighborhood $K_0$ of $0 \in \n{B}_\Gamma$ such that 
\begin{equation}\label{eq:iso_001}
\bb{V}_\Gamma|_{K_0}\; \simeq\; C(K_0,\bb{A}_\Gamma)\;,
\end{equation}
 where $\bb{V}_\Gamma|_{K_0}: = \bb{V}_\Gamma / \bb{I}_{K_0}$ and 
$$
\bb{I}_{K_0}\;: =\: \Xi(C_0(\n{B}_\Gamma\setminus K_0)) \bb{V}_\Gamma \;=\: \bigcap_{\kappa'\in K_0} \bb{I}_{\kappa'}\;.
$$  
Let assume for the moment the validity of \eqref{eq:iso_001}. One may define $K_\kappa:= K_0+\kappa$ to cover any other point $\kappa\in\n{B}_\Gamma$. 
By Proposition \ref{prop:ideal_structure_Ik}, the $*$-automorphism $\theta_\kappa$ of $\bb{V}_\Gamma$ sends $\bb{I}_{K_0}$ to $\bb{I}_{K_\kappa}$ isomorphically, inducing the isomorphism
\[
\bb{V}_\Gamma|_{K_\kappa}\; \simeq\; \bb{V}_\Gamma|_{K_0}\; \simeq\; C(K_0,\bb{A}_\Gamma) \;\simeq\; C(K_\kappa,\bb{A}_\Gamma) \;,
\]
which implies the local triviality of $\bb{V}_\Gamma$.

Fix $0<\epsilon<1/4$ and define the set 
$$
Q_0 \;:=\;  \left\{ y_1\rr{f}^1+\ldots+y_d\rr{f}^d \in\R^d\;|\; y_1,\ldots,y_d\in(-\epsilon,\epsilon)\right\} 
$$ 
where $\rr{f}^1,\dots,\rr{f}^d$ are the generators of $\Gamma'$. 
In other words $Q_0\simeq (-\epsilon,\epsilon)^d$ is 
the open rescaled cell  centered at $0$ with   sides of length $2\epsilon$. 
Let $O_{0}\subset \n{B}_\Gamma$ be  the set of points $\kappa'\in\n{B}_\Gamma$ that are the images of points in $Q_0$ under the quotient $Q_0\subset\R^d\to\R^d/\Gamma'\simeq \n{B}_\Gamma$. Clearly $O_{0}$ is open, $K_{0}:=\overline{O_{0}}$ is compact 
and $K_{0} \subsetneq \n{B}_\Gamma$. Most importantly, $\overline{Q_0}$ does not get to span the width of a unit cell in any direction, so the boundary conditions for the quotient are not relevant and the quotient map $\overline{Q_0}\to K_{0}$ is injective. Hence, we have homeomorphisms $O_{0}\simeq Q_0 \simeq (-\epsilon,\epsilon)^d$ 
and $K_{0} \simeq \overline{Q_0} \simeq [-\epsilon,\epsilon]^d$. For $\kappa'\in O_{0}$, define $\rho_{\kappa'}^0: \bb{V}_\Gamma \to \bb{A}_\Gamma$ by 
\[
\rho_{\kappa'}^0(u_{\gamma'} v_\beta) \;:=\; \expo{-\ii \kappa'\cdot \beta} F_{\gamma'} S_\beta \;, \qquad  \forall \gamma'\in \Gamma',\; \beta\in\R^d\;,
\]
where $\kappa'$ is interpreted as its representative in $Q_0$. Note that this contrasts with the definition of the representations $\rho_\kappa$ in \eqref{eq:rep_rho_k} where in the same exponential we took $\kappa$ to be its representative in the unit cell $Q_{\Gamma'}$. 
The main difference is that the map 
 $Q_{\Gamma'}\ni\kappa \mapsto \rho_\kappa(u_{\gamma'}v_\beta)$ is not continuous at $0$ while $K_{0}\ni\kappa'\mapsto\rho_{\kappa'}^0(u_{\gamma'} v_\beta)$ is continuous at all points of $K_{0}$. Clearly, $\rho_{\kappa'}^0$ remains a surjective $*$-homomorphism with kernel $\bb{I}_{\kappa'}$.
Now, let us define the $*$-homomorphism ${\Phi}: \bb{V}_\Gamma \to C(K_0,\bb{A}_\Gamma)$ by
\[
\left[ {\Phi}(u_{\gamma'} v_\beta)\right](\kappa') \;:=\; \rho_{\kappa'}^0(u_{\gamma'} v_\beta) \;=\; \expo{-\ii \kappa'\cdot \beta} F_{\gamma'} S_\beta \;.
\]
With the same argument as in the proof of Proposition \ref{prop:typ-fib} one can show that
\[
\ker({\Phi}) \;=\; \bigcap_{\kappa'\in K_0} \bb{I}_{\kappa'} \;=\; \bb{I}_{K_0}\;.
\]
Therefore, invoking the universal property of the quotient, one can define an injective $*$-homomorphism $\Psi: \bb{V}_\Gamma |_{K_0} \to C(K_0,\bb{A}_\Gamma)$ by the prescription $\Psi(a+\bb{I}_{K_0}): = {\Phi}(a)$ for every $a\in\bb{V}_\Gamma$. All that is needed for $\Psi$ to be a $*$-isomorphism is surjectivity. Let $G_{\gamma',\gamma,\beta}\in C(\n{B}_\Gamma,\bb{A}_\Gamma)$ be the map $G_{\gamma',\gamma,\beta}(\kappa) := \expo{-\ii \kappa\cdot\gamma} F_{\gamma'} S_\beta$. Note that both $C(\n{B}_\Gamma,\bb{A}_\Gamma)$, and  its subalgebra $C(K_0,\bb{A}_\Gamma)$, are generated by the functions  $G_{\gamma',\gamma,\beta}$. Thus, it suffices to show that there exists an element $a \in \bb{V}_\Gamma$ such that ${\Psi}(a) = G_{\gamma',\gamma,\beta}$. First note that, for any $\gamma\in \Gamma$, $\gamma'\in\Gamma'$ and $y\in Q_\Gamma$,
\[
\left[{\Psi}(v_\gamma u_{\gamma'} v_y)\right](\kappa') \;=\; \expo{-\ii \kappa\cdot\gamma} \expo{-\ii \kappa'\cdot y} F_{\gamma'} S_\beta \;=\; \expo{-\ii \kappa'\cdot y} G_{\gamma'\gamma,\beta}(\kappa') \;.
\]
where $\beta=y+\gamma$.
Let $h: \n{B}_\Gamma \to \C$ be a continuous function such that $h(\kappa') = \expo{\ii \kappa'\cdot y}$ for all $\kappa'\in K_0$. Note that we can construct such a function as  follows. Let $\rr{d} := \rr{f}^1+\dots+\rr{f}^d$ be the \virg{diagonal} vector (the sum of the basis vectors of $\Gamma'$). Since $Q_0$ fits fully within the translated unit cell $Q_{\Gamma'}-\rr{d}/2$ {(note it is now centered at $0$), one can use Tietze's extension theorem to define a function $\widetilde{h}$ taking the values $\widetilde{h}(p) = \expo{\ii p\cdot y}$ for $p\in K_0$ and $\widetilde{h}(p) = 0$ for $p\in \partial (Q_0-\rr{d}/2)$, and then extend it $\Gamma'$-periodically}. Then $h$ is just the map obtained by factoring through the quotient. With this, for any $\kappa'\in K_0$,
\[
\left[{\Psi}( \Xi(h) v_\gamma u_{\gamma'} v_y)\right](\kappa') \;=\; \expo{-\ii \kappa'\cdot y} h(\kappa') g(\kappa') \;=\; G_{\gamma'\gamma,\beta}(\kappa')\;.
\]
Consequently, $\Phi(\Xi(h) v_\gamma u_{\gamma'} v_y + \bb{I}_{K_0}) = G_{\gamma'\gamma,\beta}$. Since the collection of $G_{\gamma'\gamma,\beta}$ provides a family of generators of $C(K_0,\bb{A}_\Gamma)$, and $\Phi$ is an isometry, one concludes that $\Phi:\bb{V}_\Gamma |_{K_0} \to C(K_0,\bb{A}_\Gamma)$ is indeed surjective and thus a $*$-isomorphism.
\qed

\subsection{Bundle structure of the pure states of the \texorpdfstring{$\Gamma$}{Γ}-invariant algebra}
As a preparation for the description of the topology of Bloch-wave states, let us describe the bundle structure of the pure states of $\bb{V}_\Gamma$.
The next result  generalizes the description of the pure states of $\bb{V}_\Gamma$ presented in \cite[Proposition 3.26]{denittis-rendel-25}.

\begin{proposition}\label{prop:pure_st_gamma}
    Every element of $\s{P}_{\bb{V}_\Gamma}$ is of the form
    \begin{equation}\label{eq_st_pur_gamm_gen}
    \omega_{(\kappa,\nu)}(u_\alpha v_\beta) \;:=\;\expo{-\ii \kappa\cdot \beta } \;\nu(F_{\alpha}S_\beta)\;
    \end{equation}
    where $(\kappa,\nu)\in Q_{\Gamma'} \times \s{P}_{\bb{A}_\Gamma}$. This correspondence is bijective.
\end{proposition}
\proof
Suppose $\omega\in \s{P}_{\bb{V}_\Gamma}$ has quasi-momentum $\kappa$. Then, $\omega(v_\gamma a) = \expo{-\ii\kappa\cdot \gamma} \omega(a)$, and it follows that $\ker(\rho_\kappa) = \bb{I}_\kappa \subseteq \ker(\omega)$. From the universal property of the quotient, it follows that there exists a state $\nu \in \s{S}_{\bb{A}_\Gamma}$ such that $\omega = \nu \circ \rho_\kappa$, with $\rho_\kappa$ as in \eqref{eq:rep_rho_k}.
This implies the formula \eqref{eq_st_pur_gamm_gen}, and clearly the purity of $\omega$ implies the purity of $\nu$, so $\nu \in \s{P}_{\bb{A}_\Gamma}$. 
Let us now show that $\omega_{(\kappa,\nu)}$ is pure for any $(\kappa,\nu)\in Q_{\Gamma'} \times \s{P}_{\bb{A}_\Gamma}$. Let $(\pi_\nu,\s{H}_\nu, \psi_\nu)$ be the GNS triplet of $\nu\in \s{P}_{\bb{A}_\Gamma}$. Then $\pi_\nu$ must be irreducible, and for any $a\in \bb{V}_\Gamma$,
\[
\omega(a) \;=\; \nu\circ\rho_\kappa (a) \;=\; \langle \psi_\nu , \pi_\nu\circ \rho_\kappa(a) \psi_\nu\rangle_{\s{H}_\nu} \;.
\]
So $(\pi_\nu\circ \rho_\kappa,\s{H}_\nu, \psi_\nu)$ is a GNS triplet for $\omega$, and $\pi_\nu\circ \rho_\kappa$ is irreducible by \cite[Theorem 10.4.3]{dixmier-77}, implying that $\omega$ is pure.
\qed

\medskip

\begin{remark}[$\Gamma$-invariant pure states of $\bb{W}$]
  It is worth noting that the previous result shows that any $\Gamma$-invariant pure state of $\bb{W}$, $\omega\in \s{P}_\bb{W}^{\Gamma}$, has the form $\omega = \omega_{(\kappa,\nu)} := \nu\circ\rho_\kappa \circ \langle\cdot\rangle_\Gamma$, for some $(\kappa,\nu)\in Q_{\Gamma'} \times \s{P}_{\bb{A}_\Gamma}$, with $\langle\cdot\rangle_\Gamma$ the conditional expectation $\bb{W}\to \bb{V}_\Gamma$ described in \cite[Equation 3.8]{denittis-rendel-25}. This is seen in the generators as
\[
\omega(u_\alpha v_\beta) \;=\; \omega_{(\kappa,\nu)}(u_\alpha v_\beta) \;:=\;  \chi_{\Gamma'}(\alpha)\; \expo{-\ii \kappa\cdot \beta } \; \nu(F_{\alpha} S_\beta)\;.
\]
However, some states of the form $\omega_{(\kappa,\nu)}$ with $(\kappa,\nu)\in Q_{\Gamma'} \times \s{P}_{\bb{A}_\Gamma}$ may fail to be pure, as remarked in the proof of \cite[Proposition 3.22]{denittis-rendel-25}.
\hfill$\blacktriangleleft$
\end{remark}

\medskip

Note that the definition of the representations $\rho_\kappa$ in \eqref{eq:rep_rho_k} can be extended from $\kappa\in\n{B}_\Gamma$ to $p\in\R^d$ in a natural way. This yields a family  of representations $\rho_p$ of $\bb{V}_\Gamma$, parametrized  by $p\in\R^d$, where $\rho_p$ and $\rho_{p'}$ are unitarily equivalent if and only if $p-p'\in \Gamma'$. Then, given a pure state $\eta\in \s{P}_{\bb{A}_\Gamma}$ and a point $p\in\R^d$, one can define the state $\omega_{(p,\eta)} := \eta\circ \rho_p \in\s{P}_{\bb{V}_\Gamma}$. Let $\Gamma'$ act on $\s{P}_{\bb{A}_\Gamma}$ by $(\gamma'\cdot \eta) (F_{\zeta'} S_\beta) := \expo{\ii\gamma'\cdot\beta} \eta(F_{\zeta'} S_\beta)$. This action is well-defined since the resulting state is still $\Gamma$-periodic in $\beta$. Then, just as in the case for Bloch-wave states, one has the covariance property $\omega_{(p+\gamma',\gamma'\cdot\eta)} = \omega_{(p,\eta)}$. By equipping $\R^d\times \s{P}_{\bb{A}_\Gamma}$ with the $\Gamma'$-action $\gamma'\cdot (p,\eta) = (p+\gamma',\gamma'\cdot\eta)$, one can define 
\[
\rr{S}_1 (\bb{E}_\Gamma) \;:=\; \big(\R^d\times \s{P}_{\bb{A}_\Gamma}\big)/{\Gamma'} \;.
\] 
By construction, $\rr{S}_1 (\bb{E}_\Gamma)$ is a locally trivial fiber bundle over $\n{B}_\Gamma$ with typical fiber $\s{P}_{\bb{A}_\Gamma}$, the bundle map given by $\Gamma'\cdot(p,\eta) \mapsto \kappa_p \in \n{B}_\Gamma$. 

\begin{proposition}\label{prop:homeo_pure_bundle}
    The prescription \eqref{eq_st_pur_gamm_gen} provides an homeomorphism
    \[
    \s{P}_{\bb{V}_\Gamma} \;\simeq\; \rr{S}_1(\bb{E}_\Gamma) \;.
    \]
    In particular, the map $\rr{j}(\omega_{(\kappa,\eta)}) := \omega_{(\kappa,\eta)}\circ \Xi = \delta_\kappa$, endows $\s{P}_{\bb{V}_\Gamma}$ with the structure of a fiber bundle over $\n{B}_\Gamma$ with typical fiber $\s{P}_{\bb{A}_\Gamma}$.
\end{proposition}

\proof 
Take a convergent net $(p_i,\eta_i) \to (p,\eta)$ in $\R^d\times \s{P}_{\bb{A}_\Gamma}$. This implies that $p_i\to p$ in the usual topology of $\R^d$ and $\eta_i\to\eta$ in the weak-$\ast$ topology. Thus
\[
\omega_{(p_i,\eta_i)} (u_{\gamma'} v_\beta) \;=\; \expo{-\ii p_i\cdot \beta} \eta_i(F_{\gamma'} S_\beta) \;\to\; \expo{-\ii p\cdot \beta} \eta(F_{\gamma'} S_\beta) \;=\; \omega_{(p,\eta)}(u_{\gamma'} v_\beta)\;.
\]
By a density argument this is enough to conclude that  $(p,\eta)\mapsto \omega_{(p,\eta)}$ is continuous. Therefore, the map $\Phi:\rr{S}_1(\bb{E}_\Gamma)\to\s{P}_{\bb{V}_\Gamma}$ given by $\Phi([p,\eta]) = \omega_{(p,\eta)}$, with $[p,\eta]$ the $\Gamma'$-orbit of $(p,\eta)$, is also continuous. Moreover, $\Phi$ is bijective in view of Proposition \ref{prop:pure_st_gamma}. To prove that $\Phi$ is indeed an homeomorphism, we will extend it to a continuous bijection from a compact space to a bigger subset of states. Define $\s{F}_{\bb{V}_\Gamma}$ as the set of states of $\bb{V}_\Gamma$ of the form $\omega_{(\eta,\kappa)}:=\eta\circ \rho_\kappa$, for some $\eta\in\s{S}_{\bb{A}_\Gamma}$ and $\kappa\in\n{B}_\Gamma$. Let 
\[
\rr{S}_\s{S} (\bb{E}_\Gamma) \;:=\; \big(\R^d\times \s{S}_{\bb{A}_\Gamma}\big)/{\Gamma'} \;,
\]
defined in the same way as $\rr{S}_1 (\bb{E}_\Gamma)$, just replacing $\s{P}_{\bb{A}_\Gamma}$ by $\s{S}_{\bb{A}_\Gamma}$. This is a fiber bundle with typical fiber $\s{S}_{\bb{A}_\Gamma}$. By the same argument as above, the map $\widehat{\Phi}:\rr{S}_\s{S}(\bb{E}_\Gamma)\to\s{F}_{\bb{V}_\Gamma}$ given by $\widehat{\Phi}([p,\eta]) := \omega_{(p,\eta)}$ is continuous, and it is bijective by construction. Moreover, $\rr{S}_\s{S} (\bb{E}_\Gamma)$ is compact, since it is the image of the compact space $\overline{Q_{\Gamma'}}\times \s{S}_{\bb{A}_\Gamma}$ under the quotient map. We conclude that the mentioned map is an homeomorphism, and consequently its bijective restriction $\Phi:\rr{S}_1(\bb{E}_\Gamma)\to\s{P}_{\bb{V}_\Gamma}$ is an homeomorphism as well.
\qed

\subsection{Bloch-wave states as a fiber bundle}
The next step in our recipe is to demonstrate that that Bloch-wave states are a fiber bundle over $\n{B}_\Gamma$. While this is already contained in our previous result \cite[Theorem 5.1]{denittis-rendel-25}, here we also briefly deduce it from the results in the previous sections.

\medskip

Let $\iota_\Gamma:\s{P}_{\bb{W}}^{\Gamma}\to \s{P}_{\bb{V}_\Gamma}$ be the restriction map: $\iota_\Gamma(\omega) = \omega|_{\bb{V}_\Gamma}$. This is a topological embedding by \cite[Proposition 3.22]{denittis-rendel-25}. Therefore, the map $\pi_\Gamma:=\rr{j}\circ \iota_\Gamma:\s{P}_\bb{W}^\Gamma\to \n{B}_\Gamma$ given by $\omega_{(\kappa,\nu)}\mapsto \kappa$ endows $\s{P}_\bb{W}^\Gamma$ with the structure of a fiber bundle over $\n{B}_\Gamma$ with typical fiber $\s{P}_{\bb{A}_\Gamma}$, and $\s{P}_\bb{W}^{\Gamma,\beta}$ becomes a fiber subbundle. Indeed, from \cite[Proposition 3.22]{denittis-rendel-25} one sees that it corresponds to the fiber subbundle whose fibers are the \emph{normal} pure states of $\bb{A}_\Gamma$ (seen in its concrete representation). Also note that Lemma \ref{lemma:homeo_trace_class_sts} implies that the space of normal pure states of $\bb{A}_\Gamma$ is homeomorphic to ${^{\rm w}}\bb{G}_1(\rr{h}_\Gamma)$, the set of one-dimensional projections on $\rr{h}_\Gamma$ equipped with the weak operator topology. Therefore $\pi_\Gamma$ does indeed induce the trivial fiber bundle structure $\s{P}_\bb{W}^{\Gamma,\beta}\simeq \n{B}_\Gamma\times {^{\rm w}}\bb{G}_1(\rr{h}_\Gamma)$ described in \cite[Theorem 5.1]{denittis-rendel-25}.

\subsection{Topological phases of Bloch-wave states}

Now we are in a good position to describe the topological phases of the Bloch-wave states, which is the last  step in our scheme. The first ingredient for our description is encoded in the following observation.

\begin{lemma}\label{Lemma-sup_top}
    The following bijections of sets holds true
\[
{\Omega}(\s{P}_\bb{W}^{\Gamma,\beta}) \;\simeq\; \left[\n{B}_\Gamma,{^{\rm w}}\bb{G}_1(\rr{h}_\Gamma)\right]\;.
\]
\end{lemma}
\proof 
Note that, since $\s{P}_\bb{W}^{\Gamma,\beta}\simeq \n{B}_\Gamma\times {^{\rm w}}\bb{G}_1(\rr{h}_\Gamma)$ is a trivial fiber bundle, its sections are described (up to a homeomorphism) by functions $\kappa\mapsto (\kappa,f(\kappa))$, where $f\in C(\n{B}_\Gamma,{^{\rm w}}\bb{G}_1(\rr{h}_\Gamma))$ is any continuous function. Then, two sections $F(\kappa) = (\kappa,f(\kappa))$ and $G(\kappa) = (\kappa,g(\kappa))$ are homotopic if and only if $f$ is homotopic to $g$. Therefore 
\[
{\Omega}(\s{P}_\bb{W}^{\Gamma,\beta}) \;:=\; \left[{\rm Sec}(\n{B}_\Gamma,\s{P}_\bb{W}^{\Gamma,\beta})\right] \;\simeq\; \left[\n{B}_\Gamma,{^{\rm w}}\bb{G}_1(\rr{h}_\Gamma)\right]\;.
\]
as claimed.
\qed

\medskip

The next result provides the computation of ${\Omega}(\s{P}_\bb{W}^{\Gamma,\beta})$. Let us recall that $H^n(X,\Z)$ denotes the $n$-th cohomology group of the space $X$ with integer coefficients.
\begin{theorem}\label{theo_main_top}
The following bijections of sets hold true:
\[
{\Omega}(\s{P}_\bb{W}^{\Gamma,\beta}) 
{\;\simeq\; {\rm Vec}_\C^1(\n{B}_\Gamma)}
\;\simeq\;H^2(\n{B}_\Gamma,\Z)\;\simeq\;
\left\{
\begin{aligned}
&0&&\text{if}\;\; d=1\\
&\Z^{\oplus{\binom{d}{2}}}&&\text{if}\;\; d\geqslant 2\;.
\end{aligned}
\right.
\]
\end{theorem}
\proof
The first and second bijections follow from Lemma \ref{Lemma-sup_top} above and \cite[Corollary 1.2]{denittis-gomi-rendel-25}.
The last bijection is a standard result in cohomology which follows by simply observing that 
$\n{B}_\Gamma\simeq\n{T}^d$. 
\qed

\begin{remark}[Relation with vector bundles and $K$-theory]\label{rk_k-the_top}
The bijections in Theorem \ref{theo_main_top} are natural in the following sense. As discussed in the proof of Lemma \ref{Lemma-sup_top}, sections of $\pi_\Gamma:\s{P}_\bb{W}^{\Gamma,\beta}\to\n{B}_\Gamma$ are in bijection with continuous functions $f:\n{B}_\Gamma\to {^{\rm w}}\bb{G}_1(\rr{h}_\Gamma)$, and this bijection respects homotopy classes. Then, \cite[Theorem 1.1]{denittis-gomi-rendel-25} shows that ${^{\rm w}}\bb{G}_1(\rr{h}_\Gamma)$ is a model for the classifying space of the 1-dimensional unitary group $\n{U}(1)$, so a function $f\in C(\n{B}_\Gamma,{^{\rm w}}\bb{G}_1(\rr{h}_\Gamma))$ gives us a line bundle $f^*(\rr{p})$, corresponding to the bundle induced under $f$ of the universal bundle $\rr{p}:{\rm E}\n{U}(1)\to {^{\rm w}}\bb{G}_1(\rr{h}_\Gamma)$ of $\n{U}(1)$. 
Finally, as it is well-known, the map $f\mapsto f^*(\rr{p})$ sends homotopy equivalent functions into isomorphic line-bundles, and establishes a bijection between both, so it defines a correspondence $[f] \longleftrightarrow [f^*(\rr{p})]$, which is the bijection $[\n{B}_\Gamma,{^{\rm w}}\bb{G}_1(\rr{h}_\Gamma)]\simeq {\rm Vec}_\C^1(\n{B}_\Gamma)$.
In view of this,
the isomorphism between 
${\Omega}(\s{P}_\bb{W}^{\Gamma,\beta})$ and $H^2(\n{B}_\Gamma,\Z)$ described in
Theorem \ref{theo_main_top} can be thought of in terms of \emph{Chern classes} as in the standard theory of classification of line bundles (\cf \cite[Appendix A]{denittis-gomi-rendel-25}).
Moreover, if $1\leqslant d\leqslant 3$  it is also true that 
\[
{\Omega}(\s{P}_\bb{W}^{\Gamma,\beta})\;\simeq\;\widetilde{K}^0(\n{B}_\Gamma)
\]
where $\widetilde{K}^0$ denotes the reduced $K$-theory. This fact follows by observing that in low dimension ($1\leqslant d\leqslant 3$) every vector bundle of rank $k$ is isomorphic to the sum of a (possible) non-trivial line bundle and a trivial vector bundle of rank $k-1$. Said differently, in low dimension every vector bundle is \emph{stably equivalent} to a line bundle.
\hfill $\blacktriangleleft$
\end{remark}

\section{Topology of pure Zak states}\label{app-zac2}

Based on the results of \cite[Section 6]{denittis-rendel-25},
we will briefly discuss here some topological consequences  arising 
by selecting as sample space de set of pure Zak states $\bb{P}_\bb{W}^{\rm Z}$.

\subsection{\texorpdfstring{$C^*$}{C*}-bundle structure of the Zak subalgebra}
The Zak subalgebra is again abelian, making this case similar to that of Section \ref{sec:classif_sts_tras2}. We have that $\bb{Z}_\Gamma \simeq C(\n{B}_\Gamma\times \n{T}_\Gamma)$, so again $\bb{Z}_\Gamma$ is a trivial $C^*$-bundle over $\n{B}_\Gamma\times \n{T}_\Gamma$ with typical fiber $\C$. Its space of pure states is a trivial fiber bundle $\pi_{\rm Z}:\s{P}_{\bb{Z}_\Gamma} \to \rr{b}(\R^d)$ with a singleton as its typical fiber, with the bundle map corresponding to the homeomorphism in \cite[Proposition 6.4]{denittis-rendel-25}.

\subsection{Topological phases associated to Zak states}

As in the case of translationally invariant states, $\pi_{\rm Z}$ has a unique continuous section, corresponding to its inverse.  In turn Zak states have no non-trivial associated topological phases, \ie 
\[
\Omega(\s{P}_\bb{W}^{\rm Z}) \;:=\; \left[{\rm Sec}(\n{B}_\Gamma\times \n{T}_\Gamma,\s{P}_\bb{W}^{\rm Z})\right] \;=\;  \left[\{\pi_\tau^{-1}\}\right]\;.
\]

\subsection{\virg{Symmetry breaking} and ficticious phases}

Let us consider configurations taking values in $\s{P}_\bb{W}^{\rm Z}$. Given a quantum parameter space $X$, this is the same as studying maps $X\to \n{B}_\Gamma\times \n{T}_\Gamma$, in view of the homeomorphism $\s{P}_\bb{W}^{\rm Z}\simeq \n{B}_\Gamma\times \n{T}_\Gamma$. We can consider two relevant cases. First of all one can chose as the quantum parameter space $X$  the Pontryagin  of the full symmetry group $\Gamma\times\Gamma'$ of joint space and momentum lattice translations that is  $\n{B}_\Gamma\times \n{T}_\Gamma$. Another possibility is to choose for $X$ the Pontryagin dual $\n{B}_\Gamma$ of the symmetry subgroup $\Gamma$ of just the spatial lattice translations. Since we already know the first case is trivial, we will focus on the second one. This case does not have an immediate physical interpretation, but 
we will see that it has some intriguing mathematical content. 

\medskip

Since $\bb{Z}_\Gamma \simeq C(\n{B}_\Gamma\times\n{T}_\Gamma) \simeq C(\n{B}_\Gamma, C(\n{T}_\Gamma))$, we have $\bb{Z}_\Gamma$ is a trivial $C^*$-bundle over $\n{B}_\Gamma$ with typical fiber $C(\n{T}_\Gamma)$. As such, its set of pure states has the structure of a trivial fiber bundle over $\n{B}_\Gamma$ with typical fiber $\n{T}_\Gamma$. The bundle map $\pi_{\rm Z}: \s{P}_\bb{W}^{\rm Z} \to \n{B}_\Gamma$ is given by $\omega_{(\kappa,\nu)} \mapsto \kappa$. Indeed, the homeomorphism $\s{P}_\bb{W}^{\rm Z}\simeq \n{B}_\Gamma\times \n{T}_\Gamma$ restricts to a homeomorphism $\pi_{\rm Z}^{-1}(\{\kappa\}) := \{\omega_{(\kappa,\nu)} \;|\; \nu\in\n{T}_\Gamma \} \simeq \n{T}_\Gamma$. Let $\rm{Sec}(\n{B}_\Gamma,\s{P}_\bb{W}^{\rm Z})$ be the space of \emph{continuous sections} of this bundle, \ie, the space of continuous functions $F:\n{B}_\Gamma\to \s{P}_\bb{W}^{\rm Z}$ such that $(\pi_{\rm Z} \circ F)(\kappa) = \kappa$, for every $\kappa\in\n{B}_\Gamma$. This corresponds to the space of functions of the form $\kappa\mapsto \omega_{(\kappa,f(\kappa))}$, where $f$ is a continuous map $\n{B}_\Gamma\to\n{T}_\Gamma$.

\begin{proposition}
    The following bijections of sets hold
    \[
    \left[\rm{Sec}(\n{B}_\Gamma,\s{P}_\bb{W}^{\rm Z})\right] \;\simeq\; M_{d\times d}(\Z) \;\simeq\; \Z^{d^2} \;.
    \]
\end{proposition}
\proof
The set of homotopy classes of configurations in $\rm{Sec}(\n{B}_\Gamma,\s{P}_\bb{W}^{\rm Z})$ is in bijection with the homotopy classes of maps $\n{B}_\Gamma\to\n{T}_\Gamma$, by the correspondence mentioned above. Hence, we have bijections $\left[\rm{Sec}(\n{B}_\Gamma,\s{P}_\bb{W}^{\rm Z})\right] \simeq \left[\n{B}_\Gamma,\n{T}_\Gamma\right] \simeq \left[\n{T}^d,\n{T}^d\right]$. On the other hand, the homotopy classes of maps $\n{T}^d \to \n{T}^d$ are classified by $d\times d$ matrices with coefficients in $\Z$ (generalized degrees), so $\left[\n{T}^d,\n{T}^d\right]\simeq M_{d\times d}(\Z)$, as claimed.
\qed

\section{Beyond the pure state condition}\label{sec:degeneracy}

We can extend our methods to deal with fiber bundles of states whose fibers are not necessarily pure. This amounts, essentially, to allowing degeneracy.

\subsection{Configurations with finite degeneracy}
First let us discuss the case of finite degeneracy. Let $\bb{A}$ be any $C^*$-algebra. 
Given a state $\omega\in\s{S}_\bb{A}$, we define its Caratheodory number as
\[
\s{C}(\omega) \;:=\; \inf \left\{ k\in\N \;\left|\; \exists \omega_1,\dots,\omega_k\in\s{P}_\bb{A} \text{ such that } \omega = \sum_{j=1}^k t_j \omega_j \right.\right\} \;.
\]
This corresponds to a finite number whenever the state lies in the convex hull of the pure states (without taking closure), and is infinite otherwise. Thus, it distinguishes when a state can be decomposed into a convex combination of finitely many pure states, and if it can, yields the minimum number of pure states required to achieve such a decomposition. 
Now, for any $N\in\N$, define
\begin{align*}
    \s{D}[N]_\bb{A} \;:=\; \left\{ \omega\in\s{S}_\bb{A} \;|\; \s{C}(\omega) = N \right\} \;.
\end{align*}
That is, the elements of $\s{D}[N]_\bb{A}$ are the states $\omega$ of $\bb{A}$ such that there exist $t_1,\dots,t_N\in(0,1)$ and mutually distinct $\omega_1,\dots,\omega_N \in \s{P}_\bb{A}$ satisfying $\omega = t_1 \omega_1 +\dots + \omega_N t_N$, but cannot be decomposed in the same way into convex combinations of less than $N$ pure states. 
Note that $\s{D}[1]_\bb{A} = \s{P}_\bb{A}$.

\medskip

Back to our particular setting, we also define:
\begin{align*}
    \s{D}[N]_\bb{W}^\Gamma &\;:=\; \s{D}[N]_\bb{W} \cap \s{S}_\bb{W}^\Gamma \;,  \\
    \s{D}[N]_\bb{W}^{\Gamma,\beta} &\;:=\; \s{D}[N]_\bb{W} \cap \s{S}_\bb{W}^{\Gamma,\beta} \;.
\end{align*}
By restricting the homeomorphism in Proposition \ref{prop:homeo_mixed_bundle}, we obtain that $\s{D}[N]_\bb{W}^\Gamma$ is a locally trivial fiber bundle over $\n{B}_\Gamma$ with typical fiber $\s{D}[N]_{\bb{A}_\Gamma}$. Likewise, by Corollary \ref{cor:homeo_mixed_bundle_reg}, we obtain the following:
\begin{lemma}
    $\s{D}[N]_\bb{W}^{\Gamma,\beta}$ is a trivial fiber bundle over $\n{B}_\Gamma$, and its fibers are homotopy equivalent to ${^{\rm w}}\bb{G}_N(\rr{h}_\Gamma)$, the space of $N$-dimensional projections on $\rr{h}_\Gamma$ equipped with the weak operator topology.
\end{lemma}
\proof
By Corollary \ref{cor:homeo_mixed_bundle_reg}, the typical fiber of $\s{D}[N]_\bb{W}^{\Gamma,\beta}$ is homeomorphic to the set of operators $T\in \bb{L}^1_{+,1}(\rr{h}_\Gamma)$ of rank exactly $N$, which we denote by ${^{\rm w}}\bb{L}^1_{+,1,N}(\rr{h}_\Gamma)$. Given such an operator $T$, let $P_T$ denote the projection onto its support (\ie the orthogonal complement of its kernel). Then, the maps $h_1:{^{\rm w}}\bb{L}^1_{+,1,N}(\rr{h}_\Gamma)\to{^{\rm w}}\bb{G}_N(\rr{h}_\Gamma)$, $h_2:{^{\rm w}}\bb{G}_N(\rr{h}_\Gamma)\to{^{\rm w}}\bb{L}^1_{+,1,N}(\rr{h}_\Gamma)$ given by $h_1(T) = P_T$, and $h_2(P) = \frac{1}{N}P$ define an homotopy equivalence ${^{\rm w}}\bb{G}_N(\rr{h}_\Gamma)\approx {^{\rm w}}\bb{L}^1_{+,1,N}(\rr{h}_\Gamma)$. Indeed, $h_1\circ h_2$ is the identity on ${^{\rm w}}\bb{G}_N(\rr{h}_\Gamma)$, while $h_2\circ h_1$ is homotopic to the identity on ${^{\rm w}}\bb{L}^1_{+,1,N}(\rr{h}_\Gamma)$ via the homotopy $H(t,T) = \frac{t}{N} P_T + (1-t) T$. This can be readily checked using functional calculus.
\qed

\medskip

In the $\beta$-regular case, we can use the same arguments as in the proof of Theorem \ref{theo_main_top}, in conjunction with the above result, to provide a description of the topological phases ${\Omega}(\s{D}[N]_\bb{W}^{\Gamma,\beta}) := \big[{\rm Sec}(\n{B}_\Gamma,\s{D}[N]_\bb{W}^{\Gamma,\beta})\big]$ associated to configurations with finite degeneracy.
\begin{proposition}\label{theo_main_top-2}
The following bijection of sets holds true:
\[
{\Omega}(\s{D}[N]_\bb{W}^{\Gamma,\beta})\;\simeq\;{\rm Vec}_\C^N(\n{B}_\Gamma)\;.
\]
\end{proposition}

\begin{remark}[Classification in low dimension]\label{rk_k-the_top-2}
Let $N\in\N$ be a generic integer. When $1\leqslant d\leqslant 3$, a standard result in the theory of vector bundles provides 
${\rm Vec}_\C^1(\n{B}_\Gamma)\simeq{\rm Vec}_\C^N(\n{B}_\Gamma)$. As a consequence, a comparison with Remark \ref{rk_k-the_top} provides
\[
{\Omega}(\s{D}[N]_\bb{W}^{\Gamma,\beta})\;\simeq\;{\Omega}(\s{P}_\bb{W}^{\Gamma,\beta})\;,\qquad \text{if}\;\; 1\leqslant d\leqslant 3\;.
\]
For $d=4$ one knows that ${\rm Vec}_\C^N(\n{B}_\Gamma)$ is described by the second and fourth cohomology groups. Therefore
\[
{\Omega}(\s{D}[N]_\bb{W}^{\Gamma,\beta})\;\simeq\;H^2(\n{B}_\Gamma,\Z)\oplus H^4(\n{B}_\Gamma,\Z)\;\simeq\;\Z^{\oplus 7}
\;,\qquad \text{if}\;\;  d=4\;.
\]
In both cases one gets

\[
{\Omega}(\s{D}[N]_\bb{W}^{\Gamma,\beta})\;\simeq\;\widetilde{K}^0(\n{B}_\Gamma)\;,\qquad \text{if}\;\; 1\leqslant d\leqslant 4\;,
\]
showing the generality of $K$-theory in the classification of the topological phases encoded in $\Gamma$-invariant, $\beta$-regular states with finite degeneracy.
\hfill $\blacktriangleleft$
\end{remark}

\subsection{Infinite degeneracy and thermal configurations}

Just as we have considered sections of $\s{P}_\bb{W}^{\Gamma,\beta}$ and $\s{D}[N]_\bb{W}^{\Gamma,\beta}$, we may consider sections of the bigger fiber bundle $\s{F}_\bb{W}^{\Gamma,\beta}$ conformed by states of the form $\omega=\widetilde{\omega}\circ \langle\cdot\rangle_\Gamma$, with $\widetilde{\omega}\in \s{F}_{\bb{V}_\Gamma}^\beta$ (\cf Corollary \ref{cor:homeo_mixed_bundle_reg}), and $\langle\cdot\rangle_\Gamma$ the conditional expectation $\bb{W}\to \bb{V}_\Gamma$ described in \cite[Equation 3.8]{denittis-rendel-25}. The configurations associated to these states are called thermal configurations.

\begin{proposition}\label{theo_main_top-4}
Thermal configurations exhibit no non-trivial associated topological phases:
\[
{\Omega}(\s{F}_\bb{W}^{\Gamma,\beta})\;=\;\{\ast\}\;.
\]
\end{proposition}
\proof
Recall that from Corollary \ref{cor:homeo_mixed_bundle_reg}, one has an homeomorphism $\s{F}_{\bb{V}_\Gamma}^\beta \simeq \rr{Tr}(\bb{E}_\Gamma)$. 
On the other hand, by \cite[Proposition 3.22]{denittis-rendel-25} it holds that $\s{F}_\bb{W}^{\Gamma,\beta}\simeq \s{F}_{\bb{V}_\Gamma}^\beta$, so we have a bijection
\[
{\Omega}(\s{F}_\bb{W}^{\Gamma,\beta}) \;\simeq\; \big[{\rm Sec}(\n{B}_\Gamma,\s{F}_{\bb{V}_\Gamma}^\beta)\big] \;\simeq\; \big[{\rm Sec}(\n{B}_\Gamma,\rr{Tr}(\bb{E}_\Gamma))\big]\;.
\]
The last correspondence is given by Corollary \ref{cor:homeo_mixed_bundle_reg}. 

Equip $\R^d \times {^{\rm w}}\bb{U}(\rr{h}_\Gamma)$ with the $\Gamma'$-action given by $\gamma'\cdot(x,U) = (x+\gamma', UF_{-\gamma'})$, and let $\rr{U}(\bb{E}_\Gamma):=(\R^d \times {^{\rm w}}\bb{U}\left(\rr{h}_\Gamma)\right)/\Gamma'$. 
The first factor projection turns $\rr{U}(\bb{E}_\Gamma)$ into a principal ${^{\rm w}}\bb{U}(\rr{h}_\Gamma)$-bundle. Moreover, note that $\rr{Tr}(\bb{E}_\Gamma)$ is the ${^{\rm w}}\bb{L}^1_{+,1}(\rr{h}_\Gamma)$-bundle associated to the principal ${^{\rm w}}\bb{U}(\rr{h}_\Gamma)$-bundle $\rr{U}(\bb{E}_\Gamma)$, where $U\in{^{\rm w}}\bb{U}(\rr{h}_\Gamma)$ acts on $T\in\bb{L}^1_{+,1}(\rr{h}_\Gamma)$ by conjugation: $U\cdot T = UTU^*$. Now, since ${^{\rm w}}\bb{U}(\rr{h}_\Gamma)$ is contractible, we deduce that $\rr{U}(\bb{E}_\Gamma)$, and consequently $\rr{Tr}(\bb{E}_\Gamma)$, is trivial. This induces a correspondence
\[
{\Omega}(\s{F}_\bb{W}^{\Gamma,\beta})\;\simeq\;\big[{\rm Sec}(\n{B}_\Gamma,\rr{Tr}(\bb{E}_\Gamma))\big]  \;\simeq\; \big[\n{B}_\Gamma,{^{\rm w}}\bb{L}^1_{+,1}(\rr{h}_\Gamma)\big] \;=\;\{*\}\;,
\]
Where the last equality is due to the contractibility of ${^{\rm w}}\bb{L}^1_{+,1}(\rr{h}_\Gamma)$.
\qed

\subsection{The role of time-reversal invariance}

Identify $\bb{A}_\Gamma$ with its concrete form $\bb{C}(\rr{h}_\Gamma)\subset \bb{B}(\rr{h}_\Gamma)$ described in Section \ref{sec:classif_sts_latt2}. Let $C:\rr{h}_\Gamma\to\rr{h}_\Gamma$ be the anti-unitary involution of complex conjugation, defined by $Cf = \overline{f}$ for $f\in \rr{h}_\Gamma = L^2(\n{T}_\Gamma,\dd\nu)$. Given a section $F\in {\rm Sec}(\n{B}_\Gamma, \s{D}[N]_\bb{W}^{\Gamma,\beta})$, represented locally on $Q_{\Gamma'}$ as $F(\kappa) = (\kappa , P(\kappa))$ (via the fiber bundle isomorphism $\s{D}[N]_\bb{W}^{\Gamma,\beta} \simeq \n{B}_\Gamma \times {^{\rm w}}\bb{G}_N(\rr{h}_\Gamma))$, we will use the expression $\overline{F}$ to denote the section given locally by $\overline{F(\kappa)}:=(\kappa, C P(-\kappa) C)$. 
\begin{definition}
 Given a configuration $F\in {\rm Sec}(\n{B}_\Gamma, \s{D}[N]_\bb{W}^{\Gamma,\beta})$, the operation $F\mapsto \overline{F}$ is called \emph{time-reversal}, and $F$ is said to be \emph{time-reversal invariant} if $F = \overline{F}$.
\end{definition}

Let us denote with ${\rm Sec}^{C}(\n{B}_\Gamma, \s{D}[N]_\bb{W}^{\Gamma,\beta})$ the subset of time-reversal invariant configurations, and by $\Omega^C(\s{D}[N]_\bb{W}^{\Gamma,\beta}):=\big[ {\rm Sec}^{C}(\n{B}_\Gamma, \s{D}[N]_\bb{W}^{\Gamma,\beta})\big]$ the associated topological phases. The definition above states that these are identified by \virg{Real} sections in the sense of \cite[Section 4.3]{denittis-gomi-14}.
Let ${\rm Vec}_\rr{R}^N(\n{B}_\Gamma)$ be the set of isomorphism classes of  \virg{Real} vector-bundles with respect to the involution on $\n{B}_\Gamma$ given by $\kappa\mapsto -\kappa$. The reader can refer to  \cite{denittis-gomi-14} for more details on this subject.
The effect of the time-reversal symmetry on the classification of topological phases is the content of the last result.

\begin{proposition}\label{theo_main_top-3}
The following bijection of sets holds true
\[
\Omega^C(\s{D}[N]_\bb{W}^{\Gamma,\beta})\;\simeq\;{\rm Vec}_\rr{R}^N(\n{B}_\Gamma)\;.
\]
Moreover, in low dimension $1\leqslant d\leqslant 3$ one gets that
\[
\Omega^C(\s{D}[N]_\bb{W}^{\Gamma,\beta})\;=\;\{\ast\}
\]
is a singleton. 
\end{proposition}
\proof
By repeating the arguments in 
Lemma \ref{Lemma-sup_top} and Theorem \ref{theo_main_top} one gets that 
\[
\Omega^C(\s{D}[N]_\bb{W}^{\Gamma,\beta})\;\simeq\;[\n{B}_\Gamma, {^{\rm w}\bb{G}_1}(\rr{h}_\Gamma)]_{\rm eq}
\]
where the subscript \virg{eq} means that the homotopy is taken in the set of maps $P:\n{B}_\Gamma\to{^{\rm w}\bb{G}_1}(\rr{h}_\Gamma)$ which are constrained by the equivariance condition $P(-\kappa)=CP(\kappa)C$. Then, \cite[Theorem 4.13]{denittis-gomi-14}
provides the first bijection. The result in low dimension follows from 
the classification \cite[Theorem 1.6]{denittis-gomi-14}. 
\qed

\begin{remark}[The pure case; $N=1$]
It is worth adding that for $N=1$ one gets form \cite[Theorem 1.6]{denittis-gomi-14}
\[
\Omega^C(\s{P}_\bb{W}^{\Gamma,\beta})\;=\;\{[\omega_\ast]\}\;,
\]
independently of $d\in\N$. 
\end{remark}

\appendix

\section{\texorpdfstring{$C^*$}{C*}-bundles and states}\label{sec:sts_Cs_bundles}

Given a locally compact Hausdorff space $X$, a $C^*$-algebra $\bb{A}$ is said to be a $C(X)$-algebra if there exists a non-degenerate $*$-homomorphism $\Xi: C_0(X)\to \bb{Z(M(A))}$, called \emph{structure morphism}. Here, $\bb{M(A)}$ denotes the multiplier algebra of $\bb{A}$, while $\bb{Z(M(A))}$ denotes its center.
As above, we consider $\bb{A}$ to be unital and $X$ to be compact, so $\bb{M(A)} = \bb{A}$, $C_0(X) = C(X)$ and the structure morphism $\Xi:C(X)\to \bb{Z(A)}$ is trivially non-degenerate as long as it is unital. 

\medskip

Given $x\in X$, one may interpret $C_0(X\setminus\{x\})$ as the ideal of functions in $C(X)$ vanishing at $x$. Then, we may define the space $\bb{I}_x := \Xi(C_0(X\setminus\{x\})) \bb{A}$, which is a closed two-sided ideal and coincides with the set of elements $\{\Xi(g) a \in \bb{A} \;|\; g\in C_0(X\setminus\{x\}), a\in \bb{A} \}$ by Cohen's factorization theorem. We call the quotient $C^*$-algebra $\bb{A}_x := \bb{A}/\bb{I}_x$ the fiber of $\bb{A}$ at $x$, and denote the quotient map by $\pi_x : \bb{A} \to \bb{A}_x$. This map has the property that for any $a\in\bb{A}$ and $g\in C_0(X)$, $\pi_x(a\Xi(g)) = g(x)\pi_x(a)$, and allows us to interpret an element $a\in\bb{A}$ as a function $a:X\to\Pi_{x\in X} \bb{A}_x$ via $a(x) := \pi_x(a)$. The map $x\mapsto \|a(x)\|$ is in general upper semi-continuous \cite{dadarlat-09}, and a $C(X)$-algebra $\bb{A}$ for which it is continuous for all $a\in\bb{A}$ is called a \emph{$C^*$-bundle over $X$}. This is equivalent to $\bb{A}$ coinciding with the $C^*$-algebra of continuous sections of a continuous field of $C^*$-algebras in the sense of \cite{dixmier-77}.

\medskip

Let us focus on the case where our $C^*$-algebra of observables $\bb{A}$ is a $C^*$-bundle over the quantum parameter space $X$. This $C^*$-bundle structure can be interpreted as $\bb{A}$ containing the information of this quantum number. For each $x\in X$, the $C^*$-algebra $\bb{A}_x$ may be interpreted as the \emph{internal} observables of the elemental system localized at $x$, and the maps $\pi_x$ can be regarded as \emph{localizing homeomorphisms}. 

\medskip

We consider now the states of $\bb{A}$. 
Note that for each $x\in X$, each state $\omega_x\in \s{S}_{\bb{A}_x}$ defines a state $\omega\in \s{S}_{\bb{A}}$ via $\omega := \omega_x\circ \pi_x$. In fact, the pure states of $\s{S}_\bb{A}$ are exactly those of the aforementioned form where $\omega_x\in\s{P}_{\bb{A}_x}$ \cite[Theorem 10.4.3]{dixmier-77}. Define the \textit{lifting function} $\rr{i}: \bigsqcup_{x\in X} \s{S}_{\bb{A}_x} \to \s{S}_\bb{A}$ given on $\s{S}_{\bb{A}_x}$ by $\rr{i}(\omega_x) = \omega_x \circ \pi_x$. The previous statement on the correspondence of pure states then equates to the restriction $\rr{i}_\s{P}: \bigsqcup_{x\in X} \s{P}_{\bb{A}_x} \to \s{P}_{\bb{A}}$ of the lifting function being bijective. One can also have this type of structure for an appropriate sample state space $\s{Q}_\bb{A}$ by assuming that for each $x\in X$, there exists a subset $\s{Q}_{\bb{A}_x}\subset\s{P}_{\bb{A}_x}$ such that $\rr{i}_\s{P} \left(\bigsqcup_{x\in X} \s{Q}_{\bb{A}_x}\right) = \s{Q}_{\bb{A}}$. 
For later use, we define $\pi:\s{P}_\bb{A}\to X$ to be the map that takes a state $\omega=\omega_x\circ\pi_x$ to $x$.

\medskip

As was mentioned in the introduction, a configuration assigns a state of the system to each value $x\in X$ of a certain quantum number in a continuous manner. We assume the $C^*$-algebra $\bb{A}$ of the system contains the information of these quantum numbers, which translates into the condition of $\bb{A}$ being a $C^*$-bundle over $X$. In this case, one would expect a \virg{physical} configuration $F$ to assign to each value of $x$ a state $F(x)$ of the system in which this quantum number takes the value $x$. 
For example, consider a system comprised of one particle in three dimensions and let $X = \R^3 \cup\{\infty\}$ be the one-point compactification of $\R^3$, interpreted as the momentum space. Then for each value of momentum $p\in\R^3 \cup\{\infty\}$, $F(p)$ should be a state of the system in which it has momentum $p$.
This is made rigorous by the following definition of localizable configuration, which generalizes the one given in the introduction.

\begin{definition}[Localizable configuration]\label{def:loc_config}
    A given configuration $F\in C(X,\s{P}_\bb{A})$ is said to be \emph{localizable} if for each $x\in X$, $F(x) \in \rr{i}(\s{P}_{\bb{A}_x})$. 
\end{definition}

In many physical applications it makes sense to consider the extended system $\bb{A}$ as comprised by copies indexed by a quantum parameter of a given elemental (local) system. The internal observables of this elemental system are represented by a fixed $C^*$-algebra $\bb{O}$. This translates into the condition of the $C^*$-bundle $\bb{A}$ having typical fiber $\bb{A}_x \simeq \bb{O}$. 

\medskip

For a compact set $K\subset X$, we define the ideal of elements of $\bb{A}$ vanishing in $K$ as $\bb{I}_K := \Xi(C_0(X\setminus K)) \bb{A}$, and the observables localized at $K$ as $\bb{A}|_K := \bb{A}/\bb{I}_K$. 
We say the $C^*$-bundle $\bb{A}$ is trivial if $\bb{A}\simeq C(X,\bb{O})$, and locally trivial if there exists an open cover $\s{C}$ of $X$ by compact neighbourhoods (that is, a collection of compact sets whose interiors cover the whole space) such that for each $K\in\s{C}$, $\bb{A}|_K \simeq C(K,\bb{O})$.

\medskip

For our purposes, we need a good description of the localizable configurations and their relation to the space of pure states of $\bb{A}$. 
In \cite{denittis-25} it was shown that if $\bb{A}$ is a trivial $C^*$-bundle, then the space of pure states of $\bb{A}$ is homeomorphic to the trivial fiber bundle $X\times\s{P}_\bb{O}$. In \cite{rendel-25} this result will be extended to show that if $\bb{A}$ is any $C^*$-bundle over $X$ with faithful structure homomorphism and fiber $\bb{O}$, then, equipped with the map $\pi:\s{P}_\bb{A}\to X$ defined above, its space of pure states is a bundle over $X$ with typical fiber $\s{P}_\bb{O}$, which is a locally trivial fiber bundle with structure group $\rm{Aut}(\bb{O})$ if $\bb{A}$ is locally trivial as a $C^*$-bundle. In both cases, localizable configurations amount to sections of $\s{P}_\bb{A}$ viewed as a fiber bundle $\s{P}_\bb{O}\to \s{P}_\bb{A}\to X$. We do not need these results though, since in \cite{denittis-rendel-25} we directly proved the fiber bundle structure of the subsets of pure states (sample state spaces) of interest for this chapter. 
%

\subsection{Localizable states of \texorpdfstring{$\bb{V}_\Gamma$}{\text{V\_Γ}}}

The proof of Proposition \ref{prop:homeo_pure_bundle} can give significative insight on the structure of various families of states of $\bb{V}_\Gamma$ (and consequently of $\Gamma$-invariant states of $\bb{W}$), if one is looking for it. Borrowing its notation, we also obtain the following.

\begin{proposition}\label{prop:homeo_mixed_bundle}
    The map $\widehat{\Phi}$ described in the proof of Proposition \ref{prop:homeo_pure_bundle} is an homeomorphism 
    \[
    \s{F}_{\bb{V}_\Gamma} \;\simeq\; \rr{S}_\s{S} (\bb{E}_\Gamma) \;.
    \]
    In particular, $\s{F}_{\bb{V}_\Gamma}$ is a locally trivial fiber bundle with typical fiber $\s{S}_{\bb{A}_\Gamma}$.
\end{proposition}

\medskip

Note that any subset $\s{K}_{\bb{A}_\Gamma}\subset\s{S}_{\bb{A}_\Gamma}$ induces a subbundle of $\rr{S}_\s{S} (\bb{E}_\Gamma)$, and thus of $\s{F}_{\bb{V}_\Gamma}$. This allows us to treat sections $F:\n{B}_\Gamma \to \s{K}_{\bb{A}_\Gamma}$ of these bundles as \virg{degenerate} localizable configurations, and define an analogous notion of degenerate topological phases.

\medskip

We can also use Proposition \ref{prop:homeo_mixed_bundle} to enhance our description of the $\beta$-regular states of $\bb{V}_\Gamma$. First let us prove a preliminary result. Let $\s{S}_{\bb{A}_\Gamma}^{\beta}$ be the set of states $\nu$ of $\bb{A}_\Gamma$ such that $\beta\mapsto\nu(F_{\gamma'} S_\beta)$ is continuous for all $\gamma'\in\Gamma'$.

\medskip

\begin{lemma}\label{lemma:homeo_trace_class_sts}
    There is an affine homeomorphism 
    \[
    \s{S}_{\bb{A}_\Gamma}^{\beta} \;\simeq\; {^{\rm w}}\bb{L}^1_{+,1}(\rr{h}_\Gamma) 
    \]
    given on $T\in\bb{L}^1_{+,1}(\rr{h}_\Gamma)$ by $T\mapsto \nu_T$, with $\nu_T(a):= {\rm Tr}_{\rr{h}_\Gamma}(Ta)$.
\end{lemma}
\proof
Let us start by proving that $\s{S}_{\bb{A}_\Gamma}^{\beta}$ coincides with the set of normal states of $\bb{A}_\Gamma$, considering it in its concrete form as a subalgebra of $\bb{B}(\rr{h}_\Gamma)$. Since the normal states of $\bb{A}_\Gamma$ are exactly those of the form $\nu_T$ for some $T\in\bb{L}^1_{+,1}(\rr{h}_\Gamma)$, this would then give us the bijectivity of the mapping. 
Firstly, any normal state $\omega\in\s{S}_{\bb{A}_\Gamma}$ is in $\s{S}_{\bb{A}_\Gamma}^{\beta}$, since $\beta\mapsto {\rm Tr}_{\rr{h}_\Gamma}(TF_{\gamma'}S_\beta)$ is continuous for any $T\in\bb{L}^1_{+,1}(\rr{h}_\Gamma)$. For the other inclusion, fix $\nu\in \s{S}_{\bb{A}_\Gamma}^{\beta}$, and let $(\pi_\nu,\rr{h}_\nu,\varphi_\nu)$ be its GNS representation. Since $\nu$ is $\beta$-regular, one has that $\beta\mapsto\pi_\nu(v_\beta)$ is strongly continuous. Then, the Stone-Von Neumann-Mackey theorem \cite[Theorem 1]{mackey-49} implies that $\rr{h}_\nu = \bigoplus_{n\in I} \rr{h}_\Gamma(n)$, where $I$ is at most countable, each $\rr{h}_\Gamma(n)$ is a closed subspace fixed under $\pi_\nu(\bb{A}_\Gamma)$, and there exist unitary operators $U_n:\rr{h}_\Gamma(n) \to \rr{h}_\Gamma$ sending $\pi_\nu(a)|_{\rr{h}_\Gamma(n)}$ to $a\in\bb{A}_\Gamma$. Let $P_n \in\bb{P}(\rr{h}_\nu)$ be the orthogonal projection onto $\rr{h}_\Gamma(n)$, and $\varphi_\nu^n := U_n P_n\varphi_\nu$. Then,
\[
    \nu(a) \;=\; \langle\varphi_\nu,\pi_\nu(a)\varphi_\nu\rangle_{\rr{h}_\nu} \;=\; \sum_{n\in I} \big\langle P_n \varphi_\nu, \pi_\nu(a) P_n \varphi_\nu \big\rangle_{\rr{h}_\Gamma(n)} \;=\; \sum_{n\in I} \big\langle \varphi_\nu^n, a \varphi_\nu^n \big\rangle_{\rr{h}_\Gamma} \;.
\]
Define on $\rr{h}_\Gamma$ the linear operator $T\psi := \sum_{n\in I} \langle \varphi_\nu^n, \psi\rangle_{\rr{h}_\Gamma}\varphi_\nu^n$. Note that since 
\[
\sum_{n\in I}\|\varphi_\nu^n\|_{\rr{h}_\Gamma}^2 \;=\; \sum_{n\in I}\|P_n\varphi_\nu\|_{\rr{h}_\Gamma(n)}^2 \;=\; \|\varphi_\nu\|_{\rr{h}_\nu}^2 \;=\; 1\;,
\]
one has that $T$ is bounded, trace-class and has trace $1$. Moreover, $T$ is clearly positive since $\langle \psi, T\psi\rangle = \sum_{n\in I}|\langle\varphi_\nu^n,\psi\rangle_{\rr{h}_\Gamma}|^2 \geq 0$. Finally, we have that
\begin{align*}
    {\rm Tr}_{\rr{h}_\Gamma}(Ta) \;&=\; \sum_{\gamma'\in \Gamma'} \langle\xi_{\gamma'},Ta\xi_\gamma'\rangle_{\rr{h}_\Gamma} \;=\; \sum_{\gamma'\in \Gamma'}\sum_{n\in I} \langle\varphi_\nu^n,a \xi_{\gamma'}\rangle_{\rr{h}_\Gamma} \langle\xi_{\gamma'}, \varphi_\nu^n \rangle_{\rr{h}_\Gamma} \\
    &=\; \sum_{n\in I}\sum_{\gamma'\in \Gamma'} \langle\varphi_\nu^n,a P_{\xi_{\gamma'}} \varphi_\nu^n \rangle_{\rr{h}_\Gamma} \;=\; \sum_{n\in I} \langle\varphi_\nu^n,a \varphi_\nu^n \rangle_{\rr{h}_\Gamma} \\
    &=\; \nu(a) \;.
\end{align*}
Therefore $\nu$ is normal, and the second inclusion holds.

Now we will prove that $\nu_T\mapsto T$ is continuous. Let $\bb{L}^1(\rr{h}_\Gamma)$ denote the ideal of trace-class operators equipped with the trace norm $\|a\|_1 := {\rm Tr}_{\rr{h}}(|a|)$ (its Schatten ideal norm). By \cite[Proposition 2.4.3]{bratteli-robinson-87}, $\bb{L}^1(\rr{h}_\Gamma)$ is a Banach space and its Banach dual is $(\bb{L}^1(\rr{h}_\Gamma))^* = \bb{B}(\rr{h}_\Gamma)$, with duality pairing
\[
\langle A \,;\, T\rangle \;:=\; {\rm Tr}_{\rr{h}_\Gamma}(TA) \;, \qquad \forall A\in\bb{B}(\rr{h}_\Gamma),\, T\in\bb{L}^1(\rr{h}_\Gamma)\;,
\]
and the linear functionals of the form $A\mapsto\langle A \,;\, T\rangle$ are exactly those that are ultraweakly continuous.
Note that for $a\in\bb{A}_\Gamma$ and $T\in\bb{L}^1_{+,1}(\rr{h}_\Gamma)$, $\nu_T(a) = \langle a\,;\, T\rangle$.
It follows that
\begin{align*}
    \nu_{T_i}\to \nu_T \;&\iff\; \langle a \,;\, T_i\rangle \to \langle a \,;\, T\rangle \,,\; \forall a\in\bb{A}_\Gamma \\
    &\iff\; \langle A \,;\, T_i\rangle \to \langle A \,;\, T\rangle \,,\; \forall A\in\bb{B}(\rr{h}_\Gamma) \;,
\end{align*}
where the last equivalence is given by the ultraweak density of $\bb{A}_\Gamma$ in $\bb{B}(\rr{h}_\Gamma)$ and the fact that the states involved are ultraweakly continuous. This means that $T_i\xrightarrow{\rm wB} T$, where ${\rm wB}$ denotes the weak Banach topology, which then implies that $T_i\xrightarrow{\rm w} T$. Now let us assume that $T_i\xrightarrow{\rm w} T$ with $T_i,T\in\bb{L}^1_{+,1}(\rr{h}_\Gamma)$. Then, noting that $T_i = T_i^* = |T_i|$, $T = T^* = |T|$ and $\|T_i\|_1 = \|T\| = 1$, we can apply \cite[Theorem 2.20]{simon-05} and conclude that $\|T-T_i\|_1\to 0$. This then implies that $T_i \xrightarrow{\rm wB} T$, and thus that $\nu_{T_i}\to \nu_T$, as seen above. We conclude that the map $T\mapsto \nu_T$ is indeed an homeomorphism.
\qed

\medskip

Now let $\s{F}_{\bb{V}_\Gamma}^\beta := \s{F}_{\bb{V}_\Gamma}\cap \s{S}_{\bb{V}_\Gamma}^\beta$. Define the $\Gamma'$-action $\lambda:\Gamma'\curvearrowright \bb{A}_\Gamma$ via $\lambda_{\gamma'}(a) := F_{\gamma'} a F_{\gamma'}^*$. As an immediate consequence of Proposition \ref{prop:homeo_mixed_bundle} and Lemma \ref{lemma:homeo_trace_class_sts}, we obtain the following.

\begin{corollary}\label{cor:homeo_mixed_bundle_reg}
The map $\widehat{\Phi}$ of Proposition \ref{prop:homeo_mixed_bundle} restricts to an homeomorphism
\[
    \s{F}_{\bb{V}_\Gamma}^\beta \;\simeq\; \rr{Tr}(\bb{E}_\Gamma) \;:=\; \big(\R^d\times {^{\rm w}}\bb{L}^1_{+,1}(\rr{h}_\Gamma) \big)/{\Gamma'} \;,
\]
with the $\Gamma'$-action on the product being defined by $\gamma\cdot(x,T) = (x+\gamma, \lambda_{-\gamma'} (T))$.
\end{corollary}

\end{document}